\documentclass[a4paper,11pt]{article}

\RequirePackage{natbib}
\usepackage{amsmath}
\usepackage{amssymb}
\usepackage{latexsym}
\usepackage{ntheorem}
\usepackage{ifthen}
\usepackage{xcolor}
\usepackage{newcent}

\usepackage[margin=2.5cm]{geometry}

\theoremstyle{plain}

\theoremheaderfont{\bf}\theorembodyfont{\sl}

\newtheorem{theorem}{Theorem}[section]
\newtheorem{proposition}[theorem]{Proposition}
\newtheorem{lemma}[theorem]{Lemma}

\theoremheaderfont{\it}\theorembodyfont{\rm}

\newtheorem{remark}[theorem]{Remark}

\theoremheaderfont{\bf}\theorembodyfont{\rm}

\newtheorem{definition}[theorem]{Definition}
\newtheorem{example}[theorem]{Example}

\newtheorem{assumption}[theorem]{Assumption}

\theoremstyle{nonumberplain}

\theoremheaderfont{\bf}\theorembodyfont{\sl}

\newenvironment{proof}[1][]
{\ifthenelse{\equal{#1}{}}{\smallskip\noindent\textsl{Proof. }}{\smallskip
\noindent\textsl{Proof #1. }}}{\hfill$\Box$}

\def\Bc{{\cal B}}

\def \a{\alpha}

\def \d{\delta}

\def \l{\lambda}

\def \o{\omega}
\def \O{\Omega}

\def \e{\varepsilon}

\begin{document}

\title{Convergence of utility indifference prices to the superreplication price in a multiple-priors framework}

\author {
{Romain} {Blanchard}, E.mail~:  romblanch@hotmail.com.\\
\and
{Laurence} {Carassus}, E.mail~: laurence.carassus@devinci.fr \\
L\'{e}onard de Vinci P\^ole Universitaire, Research Center, 92 916 Paris La D\'{e}fense, France \\
and LMR, UMR 9008, Universit\'e Reims Champagne-Ardenne.\\
}

\date{}

\maketitle

\begin{abstract}
This paper formulates a  utility indifference pricing model for investors trading in a discrete time financial market under non-dominated model uncertainty.  Investor preferences are described by  {possibly random utility functions} defined on the positive axis. We prove that  when the investors's absolute risk-aversion tends to infinity, the multiple-priors utility indifference prices of a  contingent claim converge to its  multiple-priors superreplication price.
We also  revisit  the notion of certainty equivalent for multiple-priors  and establish its relation with risk aversion.
\end{abstract}

\section{Introduction}
In this paper, we  examine different definitions of prices for a contingent claim and their relation in the context of uncertainty.
Risk and uncertainty are at the heart of economic life  and modeling the way an agent will react to them  is a central thematic of  economic research (see for instance \citep{Gilboa}). Knightian uncertainty (see \citep{Kni}) means that the agent is not  certain about the choice of a given prior  modeling the outcome of a  situation. It is a kind of  ``unknown unknown\rq\rq{} in opposition to the risk where the agent is confident about her prior and only faces the randomness of the outcome, which is somehow the ``known unknown\rq{}\rq{}.  Issues related to uncertainty arise in various concrete situations in social sciences and economics, such as policy-making.  They  also affect  many aspects of modern finance such as  model risk   when  pricing and risk-managing  complex derivatives products  or  capital requirement quantification when looking at regulation for banks  and others financial entities.  As illustrated through the Ellsberg Paradox (see \citep{El61}), when facing uncertainty an agent  displays uncertainty aversion:  She tends to prefer a situation where the uncertainty is reduced.  This is the pendant of the risk aversion when the agent faces only risk.  It is well known that if one wants to represent the preferences of the agent in this context, the  axiomatic of the von Neumann and Morgenstern expected utility criterium   (see \citep{vNM}) is not verified. The Savage's extension (see \citep{Sav}), where subjective probability measures depending on each agent are introduced, does not solve this issue. Thus, in this paper, we follow the pioneering approach introduced by  \citep{Gilb}  where under  suitable axiomatic on the investor preferences,  the form of the utility functional is  a worst case expected utility: $\inf_{P \in \mathcal{Q}} E_{P} U( X)$,
where $\mathcal{Q}$ is the set of all  possible probability measures representing  the agent's  beliefs on the market model.
Somehow the larger  $\mathcal{Q}$,  the less confident the agent is in a specific model and the more she wishes to take into account as many scenarii as possible.  For example, the set $\mathcal{Q}$ may be constructed starting from a given  underlying model where  all the parameters are  not available but some of them may be
 inferred or estimated  from observable prices.  The agent  might also want to add her own belief or view on the ``correct"  value  of these parameters.
Such a worst case expected utility representation can also be used for robustness considerations when  $\mathcal{Q}$ is a set of models  resulting from  small perturbations of an initial reference model. This is  related for instance to the work of  \citep{Han01} where a term corresponding to the  relative entropy given a certain reference probability measure  is added to the utility functional.  The framework  of  \citep{Gilb}  was   extended by \citep{Macc06} who introduced a penalty term  to the utility functional. 
Finally, \citep{Vio11} represent  the preferences by a more general functional  $\inf_{P \in \mathcal{Q}} g(E_{P} U(X),P)$ where $g$ is a so-called uncertainty index reflecting the decision-maker's attitudes towards uncertainty.

An important feature is to allow the set of probability measures $\mathcal{Q}$  to be non-dominated. This means that no probability measure  determines the set of  events that can happen or not. The relevance of this idea is illustrated by the concrete  example of an underlying market model with volatility uncertainty, see \citep{AP95}, \citep{Ly95} and  \citep{EJ13}. {For a simple binomial model where the up and down multiples belong to intervals, the set of priors is non-dominated as soon as for one scenario, one of the interval is not trivial (see \cite{BC19}). However considering non-dominated models increases significantly the mathematical difficulties as  some of the classical tools of probability theory such as conditional expectation or essential supremum are ill-suited to this framework (since they are defined with respect to a given probability measure). These type of  issues have  contributed  to the development of innovative mathematical tools such as quasi-sure stochastic analysis, non-linear expectations, G-Brownian motions. On these topics,  we refer among  others to  \citep{pg11},  \citep{SoToZa11} or \citep{Coh12}.

In quantitative finance, the No-Arbitrage (NA) notion is  central to many problems and crucial when it comes to pricing questions. It asserts that starting from  zero wealth it is not possible to reach a positive one (i.e. non negative almost surely and strictly positive with a strictly positive probability). The characterisation of this condition or of the No Free Lunch condition is called the Fundamental Theorem of Asset Pricing (FTAP in short) and makes the link between  these notions  and the existence of  equivalent risk-neutral probability measures (also called martingale measures or pricing measures) which are equivalent probability measures that turn the (discounted) asset price process into a martingale. This was initially formalised in \citep{Hakr79}, \citep{HaPL81} and \citep{Kr81} while  \citep{dmw}  obtained the FTAP  in a general discrete-time setting under the NA condition. The literature on the subject is huge  and we refer to \citep{DelSch05} for a general overview.

The martingale measure is used  to price contingent claims. However, in  incomplete markets  i.e. when not all  contingent claims can be perfectly replicated by dynamic trading, the risk-neutral probability measure is not unique and this leads  to different possible evaluations for a given claim. The superreplication price is the minimum amount  needed for an agent selling a claim in order to superreplicate  it by trading in the market. This is the hedging price with no risk and to the best of our knowledge it was first introduced in \citep{Ben91} in the context of transaction costs. In complete markets the superreplication cost  is just the cash flow expectation computed under the unique martingale measure. But in incomplete markets,  the superreplication cost  is equal to the supremum of those expectations computed under the different risk-neutral probability measures. This is  the so called dual formulation of the superreplication price or  Superhedging Theorem (see for instance \citep{ElQu91} or \citep{CvKa92}). 

Naturally, all these concepts have seen a renewed interest in the context of uncertainty, see  among  others \citep{Ho982}, \citep{DaHo07},   \citep{Rie11}, \citep{CoOb11}, \citep{CoOb112}, \citep{ABPW13}, \citep{BeiHLPen13}, \citep {DolSon14}, \citep{BN},  \citep{Vio15}, \citep {Dol16}, \citep{ClassS} and \citep{Cher17}. 

One may wonder if the superreplication price is not too high to be used practically in  financial markets and if it should only be seen  as an  upper-bound for pricing issues.
On one hand, in \citep{CV17} it is proved that when the support of the
conditional law of the risky asset is bounded, the superreplication price of some
convex option is equal to the
replication price in a binomial model (see \citep{CRR79}) whose parameters are
the law support boundaries.
So, if this support is not too large, the superreplication price can be of practical use.
This type of result has been generalized in the robust setting by \citep{COW}: The mutilple-priors superreplication price corresponds to the
 uni-prior superreplication price for an extreme prior in $\mathcal{Q}$. \\
On the other hand,  the supperreplication price is sometimes too onerous:  For example  the  superreplication price of a call option may be equal to the underlying initial price  in a stochastic volatility model (see \citep{CvPT99}).
Thus it may be interesting  to consider another concept of pricing,
especially since the superreplication price does not take into account the preferences of the agents.

The reservation price (or utility indifference price)  is an alternative  approach to pricing contingent claims. In the context of quantitative finance  it was first introduced in \citep{H89} in the presence  of transaction costs. This is the minimum amount of money to be paid to an agent selling a contingent claim  such that, added to her initial capital,  her  utility when selling and hedging it by trading dynamically in the market is greater than or equal  to the one she would get without selling the product. Importantly, this notion of price allows to take into account the preferences of the agent and  allows for some  risk-seeking behavior while the superreplication price corresponds to a totally risk averse agent. Hence, the reservation price should provide a cheaper alternative to the suppereplication price.  But can it be used in practice? Consider  the case of illiquidity and basis  risk:  Options are sometimes written on  illiquid underlying assets   where a liquid market  exists in some closely related asset (as for example in commodity markets or for real options). To find an appropriate price and the best hedging strategy using only tradable assets, a widely used approach is the reservation price  (see \citep{VH02} and the reference therein). In the case of the exponential utility functions the price can be computed in a reasonably explicit form using convex duality. Moreover \citep{MM04} shows that the optimal strategy based on exponential utility maximization gives a superior hedging performance than a naive Black-Scholes strategy that assumes that the traded asset is a good proxy for the non-traded asset. However the reservation price for exponential utility functions is wealth independent, which is quite unrealistic since agents with different endowments will not have the same attitude towards risk. This is a good argument for considering other utility functions. In these cases, the reservation price is difficult to compute but it is still possible to obtain power series expansions (see \citep{MM04}).  

In this paper, we fill a gap in the literature introducing the reservation price  in the multiple-priors set-up and studying its links with the (multiple-priors) superreplication price. Our convergence result asserts that even in a multiple-priors set-up when the absolute risk aversion increases the preferences of the agent is less and less relevant for pricing issues: The preference based prices of the agent converge to the preference risk free one. Proving this, we extend an important literature  starting with   \citep{ElKRou00}  for exponential utility functions and a  Brownian
model. Then,  \citep{sixauthor} extended the result to a general semimartingale setting while a nonexponential case
was treated in \citep{bouchard-these}, but with severe restrictions on the utility functions.
The case of general utility functions was considered in \citep{CR06} and \citep{CR07b}
in discrete-time market models and in \citep{CR11} for continuous time ones.

We have chosen to  work in discrete time and to consider  utility functions defined on  the half real line rather than the whole real line. We believe that this is actually  relevant in practice as it corresponds to  situations where the agent is totally averse to bankruptcy.
We  also consider random utility functions to allow for state dependent absolute risk aversion or  random reference point.

To the best of our knowledge Theorems \ref{t2}, \ref{t2bis} 
(for non random utility functions) and \ref{t1} (for random utility functions) are the first general asymptotic results in the multiple-priors framework. We treat the case of general concave utility functions in a regular enough market (see Assumption \ref{SalphaI} and Theorem \ref{t2}) and  the case of general markets for sequence of functions which are possibly non concave but  bounded from above uniformly in $n$ (see Theorem \ref{t2bis}). Note that simultaneously \citep{Bart16} obtains some  convergence result for yet another  utility based price and agents with constant absolute risk aversion.

Even if the paper follows a long line of research, one could question  the theoretical (and practical) value of our asymptotic result. First, it proves that the superreplication price is a kind of universal price even when taking into account the preferences of the agent in the case of high absolute risk aversion. Indeed, empirical evidence related to the risk premium puzzle (see \citep{MP85}) has shown that the  risk aversion of an agent can be very high. Thus, in theses cases the superreplication price is a good approximation, even in the multiple-priors case. As already mentioned,    in some quite general cases (if the set $\mathcal{Q}$ is not too wide) it can be computed as in the uni-prior case and  used in practice. Outside these cases (very high risk aversion and ``reasonable\rq\rq{} superreplication price), taking into account the preferences of the agent is a good way to obtain a lower price (see the basis risk example and also Section \ref{example}).  

 We also revisit in a static context the notion of certainty equivalent introduced in \citep{Pr65}. We extend it  in the presence of multiple-priors and give some conditions for existence and uniqueness of certainty equivalent (see Proposition \ref{certainequivdef}). We  establish that the absolute risk aversion  allows the ranking of the multiple-priors  certainty equivalent despite  the presence of uncertainty aversion (see Proposition \ref{rrandce}). This part is  related to \citep{Bau13}  where  an alternative notion of (static) indifference  prices are introduced for non-random utility functions  under the representation of   \citep{Vio11}.


{Finally we present a detailed example where all the concepts and results of the paper are illustrated.}

We have chosen to  work under the discrete-time framework introduced in \citep{BN}.
We outline briefly in Section  \ref{setup} some of the interesting features of this framework, in particular with respect to time-consistency. To solve our problem, we use some arguments of {\citep{CR06} that are adapted to our multiple-priors  framework together with \citep[Theorems 2.2 and 2.3]{BN}. We also use  some  elements of quasi-sure stochastic analysis as developed in \citep{DM06} and \citep{DHP11}.

The article is structured as follows:  Section \ref{setup} presents  the framework and the definitions  needed in the rest of the paper. Section \ref{secmain}  presents the main theorem on the convergence of the utility indifference prices to the superreplication price for non-random utility function. This section also  revisits the link between certainty equivalent and absolute risk aversion in our set-up. Section \ref{example} proposes a detailed example illustrating our results.
The proofs  are reported in Section \ref{secproof}.
Finally,  in Section \ref{apen}   the  convergence result is extended to random utility functions.
\section{The model}
\label{setup}
This section presents our multiple-priors  framework.
\subsection{Framework overview}
\label{ts}
We fix  a time horizon $T\in \mathbb{N}$ and  introduce  a sequence $\left(\Omega_t\right)_{1 \leq t \leq T}$  of Polish spaces. Each $\O_{t+1}$ contains all possible scenarii between time $t$ and $t+1$. For some $1 \leq t \leq T$, let $\Omega^{t}:=\O_{1} \times \dots \times \O_{t}$  (with the convention that $\Omega^{0}$ is reduced to a singleton). We denote by  $\mathcal{B}(\O^{t})$  its Borel sigma-algebra and by $\mathfrak{P}(\O^{t})$ the set of all probability measures on $(\O^{t},\mathcal{B} (\O^{t}))$.  An element of $\Omega^{t}$ will be denoted by $\o^{t}=(\omega_{1},\dots, \omega_{t})$ for $(\o_{1},\dots,\o_{t}) \in \Omega_{1}\times\dots\times\Omega_{t}$. We also introduce the universal sigma-algebra $\mathcal{B}_{c}(\O^{t})$ which is the intersection of all possible completions of $\mathcal{B}(\O^{t})$.
A function $f: \O^{t} \to Y$ (where $Y$ is an other Polish space) is universally-measurable  (resp. Borel-measurable) if for all $B \in \mathcal{B}(Y)$ (the Borel sigma-algebra on $Y$), $f^{-1}(B) \in  \mathcal{B}_{c}(\O^{t})$ (resp. $f^{-1}(B) \in  \mathcal{B}(\O^{t})$).  Similarly we will speak of universally-adapted or universally-predictable (resp. Borel-adapted or Borel-predictable) processes.
\subsubsection{Uncertainty modelisation}
Uncertainty is modeled as in  \citep{BN} (see also  \citep{BC16}). We refer to these papers for a thorough technical presentation of the framework. 
For each $0\leq t\leq T-1$ we consider some   random set  $\mathcal{Q}_{t+1} : \Omega^t \twoheadrightarrow \mathfrak{P}(\O_{t+1})$, where $\mathcal{Q}_{t+1}(\o^{t})$ can be understood as  the set of all possible models, from the agent perspective, for the $t+1$-th period if the scenario $\o^{t}$ occurs until time $t$. Those sets  are the blocks  from which the set $\mathcal{Q}^{t}$ of all priors  on $\O^t$ is built.
\begin{align}
\label{Qstar}
 \mathcal{Q}^{t}:=\{ Q_{1}\otimes q_{2} \otimes \dots \otimes q_{t},\;  Q_{1} \in \mathcal{Q}_{1},
  q_{s+1}(\cdot,\o^s) \in \mathcal{Q}_{t+1}(\o^s),\; \forall \o^{s} \in \O^{s},\;  s \in \{1,\dots, t-1\}  \},
\end{align}
where for all $1\leq s \leq T-1$, $ q_{s+1}(\cdot,\o^s)$ is a universally-measurable stochastic kernel on $\O_{s+1}$ given $\o^s \in \O^{s}$ (see \citep[Definition 7.12 p134]{BS}), and where  the notation $Q_s:=Q_{1}\otimes q_{2} \otimes \dots \otimes q_{s}$ stands for the probability measure  resulting from the composition using Fubini's Theorem: For all $A \in \mathcal{B}_{c}(\O^{s})$
$$Q_s(A)=\int_{\Omega_{1}} \cdots  \int_{\Omega_{s}} 1_{A}(\omega_{1},\cdots, \omega_{s}) q_{s}(d\omega_s,(\omega_{s-1},\cdots \o_{1})) \cdots q_{2}(d\o_{2},\o_{1}) Q_{1}(d\omega_1).$$
The set $\mathcal{Q}^{T}$   governs the market until time $T$ and  determines which events are relevant or not for the agents.
Note that to  make this construction mathematically rigorous,  
a measurable selection theorem is applied in order to pick up some universally-measurable $q_{s+1}(\cdot,\o^s) \in \mathcal{Q}_{t+1}(\o^s)$.
To do that we rely on the following technical assumption which is now classical in the recent literature on multiple-priors models.
 \begin{assumption}
\label{Qanalytic}
For all $0\leq t\leq T-1$,  $\mathcal{Q}_{t+1}$ is a non-empty and convex valued random set such that
$$\mbox{Graph}(\mathcal{Q}_{t+1}):=\left\{(\omega^{t},P) \in \Omega^{t}\times \mathfrak{P}(\Omega_{t+1}),\; P \in \mathcal{Q}_{t+1}(\omega^{t})\right\}$$ is an analytic set. \footnote{Recall that an analytic set is the continuous image of some Polish space, see  \citep[Theorem 12.24 p447]{Hitch}, and also \citep[Chapter 7]{BS}  for more details on analytic sets.}
\end{assumption}

Apart from Assumption \ref{Qanalytic}, no specific assumption on the set of priors is made:  $\mathcal{Q}^{T}$  is neither assumed to be  dominated by a given reference probability measure nor to be weakly compact.
This setting allows for various general definitions of the  sets $\mathcal{Q}^T$.   We propose below some examples and  refer to \citep{Bart16} and \citep{BC19} for other examples.
\begin{example}
\label{exmod}
 Our model includes the (non-dominated) case where the physical measure is not known a priori but is rather a result of collecting data and estimation, so that some kind of
``neighborhood'' is added to the estimator $P^*=P^*_{1} \otimes p^{*}_{2} \otimes \cdots \otimes p^*_{T}$ of the physical measure
$$\mathcal{Q}_{t+1}(\omega^{t})=\left\{P \in \mathfrak{P}(\O_{t+1}), \, \mbox{dist}(P,p^*_t(\o^t)) \leq \e_t(\o^t) \right\},$$
where dist is  a distance between distributions (a popular choice is the Wasserstein  distance). Then if $\e_t$, $p_t^*$ and dist are Borel-measurable, $\mbox{Graph}(\mathcal{Q}_{t+1}) $ is  an   analytic set (see
\citep[Example 2.3]{Bart16}).
Another interesting non-dominated case  is when the increments of the price process are bounded. Working on the canonical space of a one-dimensional stock $\Omega= \mathbb{R}^T$, $S_t(\o^t)=\o_t$,  the set of possible priors is given by
\begin{align*}
\mathcal{Q}_{t+1}(\omega^{t}) = & \left\{P \in \mathfrak{P}(\Bc(\mathbb{R})), \, \mbox{supp}(P) \subset [\o_t d_{t+1},\o_t  u_{t+1}]  \right\},
\end{align*}
where $\mbox{supp}(P)$ is the support of the measure $P$. Then  Assumption \ref{Qanalytic} is satisfied  (see \citep{CV17}). \end{example}

\subsubsection{Time-consistency and related comments}
The fact that our sets of probability measures are uniquely determined by the set of one-step-ahead probability measures is related to the notion of time-consistency on which we  focus  briefly  now.  Roughly speaking, time-consistency means that a decision taken tomorrow will satisfy today's objective. 
 Recall that this issue appears already in a uni-prior setting, in the study of dynamic risk measures for instance, and  is  linked  to   the law of iterated conditional expectations and the dynamic programming principle.  We refer to the surveys \citep{AcPe11} and \citep{Bie16}}  for  detailed overviews.
When introducing multiple-priors  one has to be even more careful with time-consistency. In \citep[Appendix D]{Rie09}  a simple example  illustrates what can happen  if  one is not cautious  on the structure of the initial set of priors: One cannot hope to find an optimal solution using the  dynamic programming principle. 
To deal with this, one has to assume that the set of priors is stable under pasting which  roughly means  that  different priors  can be mixed  together (see  \citep[Assumption 4]{Rie09}). It is clear that the set of priors $\mathcal{Q}^{T}$ is stable under pasting. Indeed, given  \eqref{Qstar}, if $Q^1,Q^2 \in \mathcal{Q}^{T}$ with $Q^1=Q^1_{1} \otimes q^{1}_{2} \otimes \cdots \otimes q^1_{T}$, $Q^2=Q^2_{1} \otimes q^{2}_{2} \otimes  \cdots \otimes q^2_{T}$, then $R:=Q^1_{1} \otimes q^{1}_{2} \cdots \otimes q^1_{t-1} \otimes q^2_{t} \otimes \cdots \otimes q^2_{T} \in \mathcal{Q}^{T}$ for all $2 \leq t \leq T-1$. In a sense, the set $\mathcal{Q}^{T}$ is large enough (unlike in the example  considered in  \citep[Appendix D]{Rie09}). In \citep [Definition 3.1]{ES03} the equivalent notion of rectangularity is  introduced (see also  \citep[Sections 3, 4]{ES03} for more details and a graphical interpretation). \\

\subsubsection{{The traded assets and the trading strategies}}
Let $S:=\left\{S_{t},\ 0\leq t\leq T\right\}$ be a universally-adapted $d$-dimensional  process
where for $0\leq t\leq T$, $S_{t}=\left(S^i_t\right)_{1 \leq i \leq d}$ represents the  price of $d$ risky securities in the
financial market in consideration. To solve measurability issues  the following assumption already present in \citep{BN} is made.
\begin{assumption}
\label{Sass}
The price process $S$ is Borel-adapted.
 \end{assumption}
 Trading
strategies are given by universally-adapted $d$-dimensional  processes $\phi:=\{ \phi_{t}, 1 \leq t \leq T\}$ where for all $1 \leq t \leq T$, $\phi_{t}=\left(\phi^{i}_{t}\right)_{1 \leq i \leq d}$ represents the
investor's holdings in  each of the $d$ assets at time $t$. The set of trading strategies is denoted by $\Phi$.\\
The trading is assumed to be self-financed and  the riskless asset's price constant equal to $1$. The value at time $t$ of a portfolio $\phi$ starting from
initial capital $x\in\mathbb{R}$ is then given by
$$
V^{x,\phi}_t=x+\sum_{s=1}^t  \phi_s \Delta S_s.
$$

\subsection{{Multiple-priors no-arbitrage condition}}
\label{BNexp}
 As already eluded to in the introduction, the issue of no-arbitrage in the context of uncertainty has seen a renewed interest. In this paper we follow the definition introduced by \citep{BN}. 
\begin{assumption}
\label{NAQT}
The $NA(\mathcal{Q}^{T})$ condition holds true if 
$V_{T}^{0,\phi} \geq 0 \; \mathcal{Q}^{T}\mbox{-q.s.}$ for some $\phi  \in \Phi$ implies that $V_{T}^{0,\phi}  = 0 \;\mathcal{Q}^{T}\mbox{-q.s. }$
\end{assumption}
A set $N \subset X$ is a $\mathcal{Q}^{T}$-polar set  if for all $P \in \mathcal{Q}^{T}$, there exists some $A_{P} \in \mathcal{B}(\O^{T})$ such that $P(A_{P})=0$ and $N \subset A_{P}$. A property holds true $\mathcal{Q}^{T}$-quasi-surely (q.s.), if it is true outside a $\mathcal{Q}^{T}$-polar set. Finally  a set is of $\mathcal{Q}^{T}$-full measure  if its complement is a $\mathcal{Q}^{T}$-polar set.

We outline briefly some of the interesting features of this definition. First it is  a  natural and intuitive extension of the classical uni-prior no-arbitrage condition. This argument  is strengthened by the FTAP generalisation proved by \citep{BN}. Under Assumptions \ref{Qanalytic} and \ref{Sass}, the $NA(\mathcal{Q}^{T})$ is equivalent to the following: For all $Q \in \mathcal{Q}^{T}$, there exists some $P \in \mathcal{R}^{T}$ such that $ Q \ll P$ where \begin{align}
 \label{mathR}
 \mathcal{R}^{T}:=\{P \in \mathfrak{P}(\O^{T}),\; \exists \, Q^{'} \in \mathcal{Q}^{T}, P \ll Q^{'} \; \mbox{and $P$ is a martingale measure}\}.
 \end{align} The classical notion of equivalent martingale measures is replaced by the fact that for all priors $Q \in \mathcal{Q}^{T}$, there exists a martingale measure $P$ such that $Q$ is absolutely continuous with respect to $P$ and one can find an other prior  $Q' \in \mathcal{Q}^{T}$ such that $P$ is absolutely continuous with respect to $Q'$.   The extension in the same multiple-priors  setting of  the Superhedging Theorem   and subsequent results on  worst-case expected utility maximisation (see for example \citep{Nutz}) 
is an other convincing element.

We present now an alternative characterisation of the $NA(\mathcal{Q}^{T})$ condition which was proved in \citep[Proposition 2.3]{BC16}.
\begin{proposition}
\label{alpha}
Assume that the $NA(\mathcal{Q}^{T})$ condition and  Assumptions \ref{Qanalytic}, \ref{Sass} hold true. Then for all $0\leq t\leq T-1$, there exists some $\mathcal{Q}^{t}$-full measure set $\Omega^{t}_{NA} \in \mathcal{B}_{c}(\O^{t})$   such that for all $\omega^{t} \in \Omega^{t}_{NA}$,
there exists $\alpha_{t}(\omega^{t})>0$ such that for all  $h \in  D^{t+1}(\omega^{t})$  there exists $P_{h} \in \mathcal{Q}_{t+1}(\o^{t})$ satisfying
\begin{eqnarray}
\label{valaki}
P_{h}\left(\frac{h}{|h|}\Delta S_{t+1}(\omega^{t},.)<-\alpha_{t}(\omega^{t})\right)> \alpha_{t}(\omega^{t}),
\end{eqnarray}
where $D^{t+1}(\omega^{t})$ is the affine hull of the multiple-priors  conditional support of the price increments
\begin{small}
\begin{eqnarray*}
{D}^{t+1}(\o^{t})&:= \mbox{Aff} \left(\bigcap  \left\{ A \subset \mathbb{R}^{d},\; \mbox{closed}, \; P_{t+1}\left(\Delta S_{t+1}(\o^{t},.) \in A\right)=1, \; \forall \,P_{t+1} \in \mathcal{Q}_{t+1}(\o^{t}) \right\} \right).
\end{eqnarray*}
\end{small}
Note that ${D}^{t+1}(\o^{t})$ is a vector space for all $\omega^{t} \in \Omega^{t}_{NA}$.
\end{proposition}
In the case where there is only one risky asset and one period,  the interpretation of \eqref{valaki} is straightforward. It simply means  that there exists a prior (i.e. some probability measure $P^{+}$) for which the price of the risky asset increases enough and an other one ($P^{-}$) for which the price decreases, i.e. $P^{\pm} \left( \mp \Delta S(\cdot)<-\alpha\right)> \alpha$ where $\alpha>0$. The number $\a$ serves as a measure of the gain/loss and of their size. In the general case  there is always a prior in which an agent is exposed to a potential loss when buying or selling some quantity of risky assets. Note that in \citep{BC19}, the equivalence between Assumption \ref{NAQT} and condition \eqref{valaki} is established.

We present now the measurability assumption needed for our convergence results (see Theorems \ref{t2} and \ref{t1} below) when the sequence of utility functions is not bounded from above (uniformly in $n$). First, we  introduce the following spaces that will be used throughout the paper
 \begin{align*}
 \mathcal{W}^{0}_{t} & :=    \left \{ X: \Omega^{t} \to \mathbb{R}\cup \{\pm \infty\} \; \mbox{$\mathcal{B}_{c}(\O^{t})$-measurable}\right\}, \\
 \mathcal{W}^{r}_{t} & :=  \left \{ X \in \mathcal{W}^{0}_{t}, \; \sup_{P \in \mathcal{Q}^{t}}E_{P} |X|^{r} <\infty \right\} \mbox{ and }
\mathcal{W}^{\infty}_{t}  := \left \{ X\in \mathcal{W}^{0}_{t}, \;\exists M \, \geq 0, \; |X| \leq M  \; \mathcal{Q}^{t}\mbox{-q.s.}\right\},
\end{align*}
recall  \eqref{Qstar} for the definition of  $\mathcal{Q}^{t}$. We also consider  $ \mathcal{W}^{0,bo}_{T}$ the set of contingent claims   bounded from below  $\mathcal{Q}^{T}$-q.s. i.e. 
$G \in  \mathcal{W}^{0,bo}_{T}$ if and only if $G \in  \mathcal{W}^{0}_{T}$  and  there exists some constant $b \geq 0$ such that $G \geq -b$ $ \mathcal{Q}^{T}$-q.s.

Note that a superscript $^{+}$ will be added for non-negative elements (it will be also used for denoting positive parts).
\begin{assumption}
\label{SalphaI}
We have that $\Delta S_{t}, \frac{1}{\alpha_{t}} \in \mathcal{W}^{r}_{t}$ for all $1 \leq t \leq T$ and $0 < r <\infty$.
\end{assumption}
In  light of Proposition \ref{alpha}, the condition $\frac{1}{\alpha_{t}} \in \mathcal{W}^{r}_{t}$
is a  kind of strong form of no-arbitrage. Note that if $\alpha_t$  is not constant,
then even in the uni-prior case  the utility maximisation problem may be ill posed (see Example 3.3 in \citep{cr}). Hence
our integrability assumption on $\frac{1}{\alpha_{t}}$ looks reasonable.  Some concrete examples where $\a_t$ is computed and Assumption \ref{SalphaI} is verified are presented in \citep{BC19}.

Assumption \ref{SalphaI} could
be weakened to the existence of the $\mathcal{W}_{t}^{N}$-th moment for $N$ large
enough but this would lead to complicated book-keeping with no essential
gain in generality, which we prefer to avoid.

Note finally that as in \citep[Propositions 14 and 15]{DHP11} one can prove   that for all $r \in [1,\infty]$,  $\mathcal{W}^{r}_{t}$ is a  Banach space (up to the usual quotient identifying two random variables that are $\mathcal{Q}^{t}$-q.s. equal) for the norm $||\cdot||_{r,t}$ where
\begin{eqnarray*}
||X||_{r,t}:=\left(\sup_{P \in \mathcal{Q}^{t}}E_{P} |X|^{r}\right)^{\frac{1}{r}}\mbox{ if $r<\infty$  and }||X||_{\infty,t}:=\inf \{ M \geq 0, X(\cdot) \leq M \; \mathcal{Q}^{t} \mbox{-q.s.}\}.
\end{eqnarray*}
We will omit the index $t$ when $t=T$.
\section{Main results}
\label{secmain}
This section contains our main results as well as  the definitions  of the superreplication  and of the  (seller) utility  indifference prices.

\subsection{Multiple-priors superreplication price}
The multiple-priors  superreplication price (a seller price) is the minimum initial amount that  an agent will ask for  delivering some contingent claim $G \in \mathcal{W}_{T}^{0}$
so that she is fully hedged at $T$ when trading in the market.
The set of strategies which dominate  $G$  $\mathcal{Q}^{T}$-q.s. starting from a given wealth $x \in \mathbb{R}$ is defined by
\begin{align}
\label{aah} \mathcal{A}(G,x):=\left\{ \phi \in \Phi,\; V_{T}^{x,\phi} \geq G \; \mathcal{Q}^{T}\mbox{q.s.} \right\}.
\end{align}
 \begin{definition}
\label{SuperDef}
Let $G \in \mathcal{W}_{T}^{0}$. The multiple-priors  superreplication  price of $G$ is defined by
\begin{eqnarray}
\label{piG}
 \pi(G):=\inf \left\{z \in \mathbb{R}, \; \mathcal{A}(G,z) \neq \emptyset \right\}
\end{eqnarray}
and $\pi(G)=+\infty$ if $\mathcal{A}(G,z) = \emptyset$ for all $z \in \mathbb{R}$.
\end{definition}
Note that $\pi(G)=+\infty$ if and only if $\mathcal{A}(G,z) = \emptyset$ for all $z \in \mathbb{R}$.

\begin{remark}
The corresponding buyer price is the multiple-priors  subreplication  price of $G$   defined by 
$\pi^{sub}(G):=\sup \left\{z \in \mathbb{R},  \; \mathcal{A}(-G,-z) \neq \emptyset \right\}$ 
and $\pi^{sub}(G)=-\infty$ if $\mathcal{A}(-G,-z) = \emptyset$ for all $z \in \mathbb{R}$.
It is clear that for $G \in \mathcal{W}_{T}^{0}$, $\pi^{sub}(G)=-\pi(-G)$.
\end{remark}

We recall now for the convenience of the reader  \citep[Theorem 2.3]{BN}.
\begin{theorem}
\label{BN2}
 Assume that Assumptions  \ref{Sass} and \ref{NAQT} hold true and let $G \in \mathcal{W}_{T}^{0}$ be  fixed. Then $\pi(G)>-\infty$ and $\mathcal{A}(G,\pi(G)) \neq \emptyset$.
\end{theorem}
\begin{remark}
\label{Theo22}
To apply \citep[Theorem 2.3]{BN}   Assumption \ref{Qanalytic}  is not needed and one may also use  a weaker form of  Assumption \ref{Sass}.
\end{remark}
As usual if the market is complete (i.e. if  any bounded contingent claim is replicable), the superreplication price is equal to the replication price. Indeed, if  $G$ is replicable,  i.e. if there exists some $x_{G} $ and some $\phi_{G} \in \Phi$  such that $G=V_{T}^{x_{G},\phi_{G}}$ $\mathcal{Q}^{T}$-q.s., then $\pi(G)=x_{G}=\pi (V_{T}^{x_{G},\phi_{G}})$. Moreover, under some measurability assumption on $G$, the Superreplication Theorem holds: $\pi(G)=\sup_{P \in \mathcal{R}^{T}} E_{P} G$, see \citep[Superhedging Theorem]{BN} and \eqref{mathR} for the definition of $\mathcal{R}^T$.

We now turn to some pricing rules which take into account the preferences of the agents.

\subsection{Utility function and utility indifference price}
In Section \ref{secmain} we focus on deterministic utility functions defined on the half-real line. An extension to random utility functions is proposed in Section \ref{apen}. In the rest of the section, we consider utility function  $U:  \mathbb{R} \rightarrow \mathbb{R}\cup \{- \infty\}$ such that $U(x)=-\infty$ if $x<0$ and such that there exists some $x_0 \in (0,\infty)$ verifying  $U(x_0)>-\infty$. Considering such functions   corresponds to the concrete situation where the investor is infinitely averse to bankruptcy. Up to a translation one may consider investor with limited credit line.
\begin{assumption}
\label{utilitydefdif}
The restriction of $U$ on $(0,\infty)$  is strictly increasing, twice continuously differentiable and we set $U(0):=\lim_{x\to 0^+} U(x)$.  
\end{assumption}
\begin{definition}
\label{RA}
For any function $U$ satisfying Assumption \ref{utilitydefdif},  the absolute risk  aversion is defined   for all $x \in  (0,+\infty)$ by
\begin{align}
\label{absorisk}
r(x) := &-\frac{U^{''}(x)}{U^{'}(x)}.
\end{align}
\end{definition}
We now turn to pricing issues. First we define some particular sets of strategies for a contingent claim $G \in  \mathcal{W}^{0}_{T}$ and some initial wealth $x \in \mathbb{R} $ (recall \eqref{aah})
\begin{align*}
\Phi(U,G,x):= & \left\{ \phi \in \Phi,\; E_{P} U^{+}(V_{T}^{x,\phi}(\cdot) -G(\cdot))<+\infty,   \forall P \in \mathcal{Q}^{T} \right\}\\
\mathcal{A}(U,G,x) :=&  \Phi(U,G,x) \cap \mathcal{A}(G,x).
\end{align*}
If $\phi \in \Phi(U,G,x)$, the integrals $E_{P} U (V_{T}^{x,\phi}(\cdot)-G(\cdot))$ are well-defined for all $P \in \mathcal{Q}^{T}$ and belong to $[-\infty,\infty)$. However, even for $x  \geq \pi(G)$, $\mathcal{A}(U,G,x)$ might be  empty. Indeed, from  Theorem \ref{BN2} there exists some $\phi \in \mathcal{A}(G,x)$,
but $\phi$ might not belong to $\Phi(U,G,x)$. 
In order to fix these integrability issues,  the following proposition   establishes that $\mathcal{A}(U,G,x)=\mathcal{A}(U,x)$ under suitable assumptions.

\begin{proposition}
\label{pivsp}
Suppose that Assumptions \ref{Qanalytic}, \ref{Sass} and \ref{NAQT} hold true.  
{Suppose that either $U$ is bounded from above or that $U$ is a concave function satisfying Assumption \ref{utilitydefdif} and that Assumption \ref{SalphaI}  holds true}.
Fix some contingent claim  $G \in \mathcal{W}_{T}^{0,bo}$ and some $x \in \mathbb{R}$.
 Then,
$\mathcal{A}(U,G,x)=\mathcal{A}(G,x)$.  If $\mathcal{A}(G,x) \neq \emptyset$, there exists some constant $M_{x} \in [0,\infty)$ such that $-\infty \leq E_{P} U \left( V_{T}^{x,\phi}(\cdot)-G(\cdot)\right) \leq M_{x}$,  for all $P \in \mathcal{Q}^{T}$ and $\phi \in \mathcal{A}( G,x)$.
\end{proposition}
Note that \eqref{aah} and Theorem \ref{BN2} imply that $\mathcal{A}(G,x) \neq \emptyset$ if and only if $x \geq \pi(G)$. \\

\begin{proof}
If $ \mathcal{A}(G,x)= \emptyset$ then  $\mathcal{A}(U,G,x)=\emptyset$. 
Suppose now that $ \mathcal{A}(G,x) \neq \emptyset$ (which is equivalent to $x \geq \pi(G)$). 
 {Assume first that $U$ is a concave function satisfying  Assumption \ref{utilitydefdif} and that  Assumption \ref{SalphaI}  holds true}.  
The proof relies on Lemma \ref{L10}. 
The monotonicity, concavity and differentiability  of $U$ imply that  for all $y \in \mathbb{R}$ and $x_0>0$ 
$U(y) \leq U(\max(y,x_0)) \leq U(x_0)+ \max(y-x_0,0) U'(x_0).$
Thus
\begin{align}
\label{unstar}
U^{+}(y)   \leq U^{+}(x_0)+ |y| U'(x_0).
\end{align}
For any $\phi \in \mathcal{A}(G,x)$ and
 $P \in \mathcal{Q}^{T}$,   using the monotonicity of $U$, the existence of some constant $b \geq 0$ such that $G \geq -b$ $\mathcal{Q}^{T}$-q.s. and Lemma \ref{L10}, we get that
\begin{eqnarray}
\nonumber
E_{P} U^{+} (V_{T}^{x,\phi}(\cdot)-G(\cdot)) & \leq  & E_{P} U^{+} (V_{T}^{x+b,\phi}(\cdot)) \\
\nonumber
& \leq & U^{+}(1) +\sup_{P \in \mathcal{Q}^{T}} E_{P} \left(\left|V_{T}^{x+b,\phi}(\cdot)\right|\right) U'(1)\\
\label{Kx} 
& \leq & U^{+}(1) + (|x|+b)  U'(1) \sup_{P \in \mathcal{Q}^{T}}E_{P}  M_{T}(\cdot)  :=M_x<\infty
\end{eqnarray}
since $M_{T} \in \mathcal{W}^{1}_{T}$ and both $U(1)$ and $U'(1)$ are finite valued. 

 {Assume now that $U$ is bounded from above. Then the last boundness result in \eqref{Kx} is still valid choosing for $M_x$ the upper bound of $U^+$.}
\end{proof}\\

We now introduce the quantity  $u(G,x)$ which represents,  the maximum worst-case expected utility starting from initial capital $x \in \mathbb{R}$ and  delivering $G \in  \mathcal{W}^{0}_{T}$  at $T$
\begin{align}
\label{theeq}
u(G,x):= \sup_{\phi \in \mathcal{A}(U,G,x)} \inf_{P \in \mathcal{Q}^{T}} E_{P} U \left(V_{T}^{x,\phi}(\cdot)-G(\cdot)\right)\end{align}
where  $u(G,x)=-\infty$ if $ \mathcal{A}(U,G,x)=\emptyset$. Note that the expectations in  $u(G,x)$ are well-defined (by definition of $\mathcal{A}(U,G,x)$) and without further assumptions  $u(G,x) \in [-\infty,\infty]$.
\begin{remark}
\label{optimals}
The aim of  this paper is not to find an optimal solution for \eqref{theeq}.  As already mentioned, without further assumption $  \mathcal{A}(U,G,x)$ might be empty and in this case $u(G,x)=-\infty$: The agent will never sell the contingent claims and \eqref{theeq} is not really an interesting problem. If $x \geq \pi(G)$, under the assumptions of Proposition \ref{pivsp},  we can rely on  Theorem \ref{BN2}  to find strategies in  $\mathcal{A}(U,G,x)$ since any super-replication strategy (which exists)  automatically belongs to $\mathcal{A}(U,G,x)$. In this case,  we have that $u(G,x) \leq M_{x}<\infty,$ see \eqref{Kx}.  
Note that if $  \mathcal{A}(U,G,x)$ is not empty, finding  an optimal  solution for \eqref{theeq} requires  to study how the terminal constraint $V_{T}^{x,\phi} -G \geq 0$ $\mathcal{Q}^{T}$-q.s. together with the integrability condition on $E_{P} U^{+} (V_{T}^{x,\phi}(\cdot)-G(\cdot))$  are ``propagated" to the previous periods through dynamic programming. This is a non-trivial problem. This was done for $G=0$ in \citep{BC16} and  could be done in the  general case along the lines of \citep[Remark 5]{RS06}, see also \citep[Theorem 2.3, Lemma 4.10]{BN}.  The problem has been solved  for bounded from above utility function in  \citep[Theorem 4.3, Lemma C.8]{COW}.
\end{remark}

We are now in a position to define  the  (seller) multiple-priors  utility indifference  price or reservation price, which generalizes in the presence of uncertainty the concept  introduced by \citep{H89}. It represents the minimum amount of money to be paid to an agent selling a contingent claim $G$ so that added to her initial capital,  her multiple-priors utility when selling $G$ and hedging it by trading dynamically in the market is greater than or equal to  the one she would get without selling this product.
\begin{definition}
\label{UtIndP}
Let $G \in \mathcal{W}^{0}_{T}$ and $x \in \mathbb{R}$.  The (seller) multiple-priors  utility indifference  price  is  defined by
\begin{align}
\label{prixut}
p(G,x):= \inf \left\{z \in \mathbb{R},\; u(G,x+z) \geq u(0,x)\right\}
\end{align}
where $p(G,x)= +\infty$ if  $u(G,x+z) < u(0,x)$ for all $z \in \mathbb{R}$.
\end{definition}
{
\begin{remark}
 \label{wealth}
In the previous definition $x$ represents the initial wealth of the agent. It could be the accumulated trading gain   and hence much higher than the price of a given contingent claim. However,  nothing prevents $x$ from being smaller than the superreplication price, see Section \ref{example}.   
\end{remark}}
\begin{remark}

The utility indifference price  in \eqref{prixut} is static. Nevertheless  we believe it is important to allow for a dynamic evolution of the price process $S$ and of the uncertainty on the priors between the initial date  and the maturity. Moreover, it could be proved that $p(G,x)$ coincides at $t=0$ with the dynamic version of the utility indifference price. 
\end{remark}
\begin{remark}
One may also introduce the  (buyer) multiple-priors  utility indifference  price $p^{B}(G,x):= \sup \left\{z \in \mathbb{R},\; u(-G,x-z) \geq u(0,x)\right\}$ 
where $p^{B}(G,x)= -\infty$ if $u(-G,x-z) < u(0,x)$ for all $z \in \mathbb{R}$. If $G \in \mathcal{W}^{\infty}_{T}$ then under the assumptions of Proposition \ref{pivsp}, $\mathcal{A}(-G,z)=\mathcal{A}(U,-G,z)$ for all $z \in \mathbb{R}$ and   $p^{B}(G,x)=-p(-G,x)$. 
\end{remark}
The next proposition shows that under suitable assumptions whatever the preferences of the agent are, the reservation price  is lower than the superreplication price.  The superreplication price is, in the sense that we will define below, the price corresponding to an infinite absolute risk averse agent.
\begin{proposition}
\label{inf}
Assume that Assumptions \ref{Qanalytic}, \ref{Sass} and \ref{NAQT} hold true. 
{Suppose that either $U$ is a non decreasing and bounded from above function or that $U$ is a concave function verifying Assumption \ref{utilitydefdif} and that Assumption \ref{SalphaI}  holds true}.
Fix some contingent claim  $G \in \mathcal{W}_{T}^{0,bo}$ and some $x \in \mathbb{R}$.
 Let $x \in \mathbb{R}.$ Then {
  $p(G,x)\leq \pi(G)$.  Assume that  $u(0,x)>-\infty$. Then   $p(G,x)\geq \pi(G)-x$ and $\lim_{x \to 0} p(G,x)=\pi(G)$. }

 \end{proposition}
\begin{proof}
Fix some $x  \in \mathbb{R}$. First Theorem \ref{BN2} gives some $\phi_{G} \in\mathcal{A}(G,\pi(G))$. As $U$ is non-decreasing, we get that
 \begin{align*}
u(0,x) = \sup_{\phi \in \mathcal{A}(U,0,x)} \inf_{P \in \mathcal{Q}^{T}} E_{P} U( V_{T}^{x,\phi}(\cdot))
&\leq  \sup_{\phi \in \mathcal{A}(U,0,x)} \inf_{P \in \mathcal{Q}^{T}} E_{P} U\left(V_{T}^{x+\pi(G), \phi+ \phi_{G}}(\cdot) - G(\cdot)\right)\\
&\leq  \sup_{\phi \in \mathcal{A}(U,G,x+ \pi(G))} \inf_{P \in \mathcal{Q}^{T}} E_{P} U\left(V_{T}^{x+\pi(G), \phi}(\cdot) - G(\cdot)\right)\\
&=u(G,x+\pi(G)),
\end{align*}
where Proposition \ref{pivsp} has been used for the second inequality~: If $\phi \in \mathcal{A}(U,0,x) = \mathcal{A}(0,x),$ 
then $\phi+\phi_{G} \in \mathcal{A}(G,\pi(G)+x)=
\mathcal{A}(U,G,\pi(G)+x)$.
So $p(G,x) \leq \pi(G)$ follows from \eqref{prixut}.  \\
{Assume now that  $u(0,x)>-\infty$.  If $x+z< \pi(G)$, then  $u(G,x+z) =-\infty< u(0,x),$ which implies that $p(G,x)\geq \pi(G)-x$ and  $\lim_{x \to 0} p(G,x)=\pi(G)$ follows.}
\end{proof}

\subsection{Asymptotic result}
\label{priceconvergence}
Intuitively speaking an agent who is totally risk averse will use the superreplication price: Whatever the possible outcomes
(where possible outcomes  are defined by  a set of probability measures), she does not want to incur any loss (see \eqref{piG}). We  prove now that  the utility indifference price goes to the superreplication price when the absolute risk aversion goes to infinity. We will  see also  in Proposition \ref{rrandce} that the absolute risk aversion remains a good tool to rank a proper notion of certainty equivalent in the presence of multiple-priors. The result of \citep{Pr65} remains true: Increasing absolute risk-aversion implies decreasing  certainty equivalent. \\

\noindent We consider some contingent claim $G \in \mathcal{W}^{0,bo}_{T}$ and a sequence of non-random  functions  $\left(U_{n}\right)_{n \geq 1}$ where for all $n \geq 1$, $U_n:  \mathbb{R} \rightarrow \mathbb{R}\cup \{- \infty\}$  is such that $U_n(x)=-\infty$ if $x<0$ and such that there exists some $x_n \in (0,\infty)$ verifying $U_n(x_n)>-\infty$. For all $n\geq 1$, we denote by 
$u_{n}(G,x)$ the value function in \eqref{theeq} for $U_n$, $p_{n}(G,x)$ the indifference price for $U_n$ (see \eqref{prixut}) and $r_{n}$ the absolute risk aversion for $U_n$ (see \eqref{absorisk}) when $U_n$ satisfies Assumption  
\ref{utilitydefdif}. 
\begin{theorem}
\label{t2}
Assume that Assumptions \ref{Qanalytic}, \ref{Sass}, \ref{NAQT}  and \ref{SalphaI} hold true.  Let $(U_n)_{n\geq 1}$ be a sequence  of  {concave}  utility functions satisfying Assumption \ref{utilitydefdif} such that  $\lim_{n \to +\infty} r_n(x)=+\infty$ for all $x>0.$ Then  
$\lim_{n \to+ \infty} p_{n}(G,x)=\pi(G)$ for all $x>0$ and $G \in  \mathcal{W}^{0,bo}_{T}.$
\end{theorem}
\begin{theorem}
\label{t2bis}{
Assume that Assumptions \ref{Qanalytic}, \ref{Sass} and \ref{NAQT}  hold true.  Let $G \in  \mathcal{W}^{0,bo}_{T}$ and $x_0 > 0.$ The  following two assertions hold true. \\
1. Let $(U_n)_{n\geq 1}$ be a sequence  of concave functions satisfying Assumption \ref{utilitydefdif}. Assume that the sequence $(\frac{U_n}{U_n\rq{}(x_0)}-\frac{U_n(x_0)}{U_n\rq{}(x_0)})_{n\geq 1}$ is bounded from above uniformly in $n$ for $n$ big enough and that  $\lim_{n \to +\infty} r_n(x)=+\infty$ for all $0<x<x_0$. Then $\lim_{n \to+ \infty} p_{n}(G,x_0)=\pi(G).$ \\
2. Let $(U_n)_{n\geq 1}$ be a sequence  of non decreasing functions which are bounded from above uniformly in $n$ for $n>N$ for some $N\geq 1.$ Assume that $\inf_{n\geq N}U_n(x_0)>-\infty$  and that  $\lim_{n \to +\infty} U_n(x)=-\infty$ for all $0<x<x_0.$ Then 
$\lim_{n \to+ \infty} p_{n}(G,x)=\pi(G)$.}
\end{theorem}

\begin{remark}
\label{rembound1}{
We comment on Theorem \ref{t2bis}. For the first item, note that if for some $N>1$, the $(U_n)_{n\geq N}$  are bounded from above uniformly in $n,$  
$\inf_{n\geq N} U_n(x_0) >-\infty$ and $\inf_{n\geq N} U_n\rq{}(x_0) >0$ then the sequence  $(\frac{U_n}{U_n\rq{}(x_0)}-\frac{U_n(x_0)}{U_n\rq{}(x_0)})_{n\geq N}$ is bounded from above uniformly in $n$. This is not a necessary condition. Let $U_{n}(x)=-e^{-r_{n}x}$  for $x\geq 0$ and $U_{n}(x)=-\infty$ otherwise where $\lim_{n \to +\infty} r_{n}=+\infty$.  Then for all $N \geq 1$ and all $x_0>0$, $\inf_{n\geq N} U_n\rq{}(x_0) =0$. Nevertheless, for all $x_0>0$,
$\frac{U_n (x)}{U_n\rq{}(x_0)}-\frac{U_n(x_0)}{U_n\rq{}(x_0)}=-\frac1{r_n}e^{-r_n(x-x_0)} + \frac1{r_n}$ is bounded from above uniformly in $n$ for $n$ big enough. \\
Finally, if $\inf_{n\geq N} U_n\rq{}(x_0) >0,$ a straightforward adaptation of   \citep[Lemma 4]{CR06} shows that 
$\lim_{n \to +\infty} U_n(x)=-\infty$ for all $0<x<x_0$. But  if this last condition holds true, the  convergence result does not even need  the functions to be either concave or differentiable, see item 2.}
\end{remark}
\begin{proof}
The proof of Theorems \ref{t2} and \ref{t2bis} are reported in Section \ref{proofsec}.
\end{proof}
\begin{remark} 
Under the same assumptions,  Theorem \ref{t2} implies that if  $G \in  \mathcal{W}^{\infty}_{T}$ then 
 $\lim_{n \to +\infty} p^{B}_{n}(G,x)=\pi^{sub}(G)$.
 \end{remark}
\begin{remark} We give some intuition on Theorem \ref{t2}:  For a utility function that has a sort of infinite absolute risk  aversion, we show that
the utility indifference price is equal to the superreplication price. Let $G \in \mathcal{W}^{0,bo}_{T}$ and assume for sake of simplicity that $\pi(G) \geq 0$.
Fix some $x \geq \pi(G)$ and introduce the following utility function $U_{\infty}: \mathbb{R} \to \mathbb{R} \cup \{-\infty\}$
where $U_{\infty}(y):= -\infty 1_{(-\infty,x)}(y)$.
The absolute risk  aversion of $U_{\infty}$ is not defined, but
  $U_{n}(y) := - e^{-n(y-x)}$ for $y \geq 0$ and $U_{n}(y):=-\infty$
for $y<0$ verifies Assumption \ref{utilitydefdif} and  for $y \geq 0$ fixed with $y \neq x$,
$\lim_{n \to +\infty} U_{n}(y)=U_{\infty}(y)$. Since the absolute risk  aversion of the utility
functions $U_{n}$ goes to $+\infty$, one may say that $U_{\infty}$ has an
infinite absolute risk  aversion.
We now show that the superreplication price of $G$ is equal to its utility indifference price evaluated with the function  $U_{\infty}$. First, it is easy to see that $u_{\infty}(0,x)=0$.
 Since for all $\phi \in \Phi$, $y \in \mathbb{R}$, $U_{\infty}^{+}(V_{T}^{y,\phi}(\cdot)-G(\cdot))=0$, we have that
  $\Phi(U_{\infty},G,y)=\Phi(U_{\infty},0,y)=\Phi$ and $\mathcal{A}(G,y)=\mathcal{A}(U_{\infty},G,y)$.  Theorem \ref{BN2} implies that
$\mathcal{A}(U_{\infty},G,y)$ is not empty for all $y \geq \pi(G)$. Now fix some $ z  < \pi(G)$ and $\phi \in \mathcal{A}(U_{\infty},G,x+z)$.
There exists some $P \in \mathcal{Q}^{T}$ such that $P(V_{T}^{z,\phi}(\cdot) -G(\cdot)<0)>0$ or equivalently $P(V_{T}^{x+z,\phi}(\cdot) -G(\cdot)<x)>0$ which implies that
$E_{P} U_{\infty}(V_{T}^{x+z,\phi}(\cdot)-G(\cdot)) =-\infty$.  Hence for all $\phi \in \mathcal{A}(U_{\infty},G,x+z)$,
$\inf_{P \in \mathcal{Q}^{T}} E_{P} U_{\infty}(V_{T}^{x+z,\phi}(\cdot)-G(\cdot)) =-\infty$ and
 $u_{\infty}(G,x+z)=-\infty<u_{\infty}(0,x)$ follows.  The definition  of $p(G,x)$ implies that $p(G,x) \geq z$ and letting $z$ go to $\pi(G)$, $p(G,x) \geq \pi(G)$. {Proposition \ref{inf} shows that the equality holds true.}
\end{remark}
\begin{remark}
The dominated case is of special interest. Assume that there exists some $P^* \in \mathcal{Q}^{T}$ such that for all $P \in \mathcal{Q}^{T}$, $P$ is absolutely continuous with respect to $P^*$. Then a set is of $\mathcal{Q}^{T}$-full measure if and only if it is of $P^*$-full measure. So $NA(\mathcal{Q}^{T})$ is equivalent to $NA(P^*)$ i.e. 
$V_{T}^{0,\phi} \geq 0 \; P^*\mbox{-a.s. }$ for $\phi  \in \Phi$ implies that $V_{T}^{0,\phi}  = 0 \;P^*\mbox{-a.s. }$ Moreover 
$$\pi(G)=\inf \left\{ \phi \in \Phi,\; V_{T}^{x,\phi} \geq G \; P^*\mbox{a.s.} \right\}:=\pi^{P^*}(G).$$
Theorem \ref{t2} rephrases as follows: For $G \in  \mathcal{W}^{0,bo}_{T}$, if Assumptions \ref{Qanalytic}, \ref{Sass}, \ref{SalphaI} and $NA(P^*)$  hold true and  if $\lim_{n \to +\infty} r_n(x)=+\infty$ for all $x>0$ then $\lim_{n \to+ \infty} p_{n}(G,x)=\pi^{P^*}(G)$ for all $x>0$. At the limit we are back to a uni-prior setup. {The same holds true for Theorem \ref{t2bis}  changing Assumption \ref{NAQT} by $NA(P^*)$.}
\end{remark}

\subsection{\textbf{Absolute risk aversion and certainty equivalent}}
\label{secrav}
In the uni-prior case, we know from \citep{Pr65} that the  absolute risk-aversion allows to rank the certainty equivalent:  An
increasing absolute risk-aversion  implies a decreasing certainty equivalent. In this section we will see how this property extend to the multiple-priors case.

Recall the uni-prior case where $\mathcal{Q}^{T}=\{P\}$ and the preferences are represented by a  utility function $U.$ For  a given    asset whose payoff at maturity is  $G$,
the certainty equivalent $e(G,P)$ is the amount of cash that will make the agent indifferent (in the sense of the expected utility
evaluation) between receiving the cash  and the asset  $G$
$$E_{P} U(e(G,P))=U(e(G,P))=  E_{P} U(G(\cdot)).$$
The risk premium
$\lambda(G,P):= E_{P} G(\cdot)-e(G,P)$ is the amount  the agent is willing to lose in order to be indifferent (in the sense of the expected utility
evaluation) between the sure quantity $E_{P} G(\cdot)-\lambda(G,P)$
and the random variable $G.$ \\
The following proposition gives the definition of the certainty equivalent in a multiple-priors  framework and   provides conditions for existence and uniqueness.
It also establishes  that under suitable assumptions $\lambda(G,P) \geq 0$. The risk premium  is thus a measure of  the risk-aversion of the agent:
The higher the risk premium, the more risk-averse the agent is.
\begin{proposition}
\label{certainequivdef} Let  $U$ be  a {concave} utility function satisfying Assumption \ref{utilitydefdif} and let
\begin{small}
\begin{eqnarray}
\label{WTU}
\mathcal{W}^{+}_{T}(U):= \left\{G \in \mathcal{W}^{0,+}_T, \, G(\cdot) <+\infty \; \mathcal{Q}^{T} \mbox{-q.s.,}\, \; E_{P} U^+(G(\cdot)) <\infty, \forall P\in \mathcal{Q}^{T},
\sup_{P \in \mathcal{Q}^{T}}E_{P} U^-(G(\cdot))<\infty \right\}
\end{eqnarray}
\end{small}
Assume that $G \in \mathcal{W}^{+}_{T}(U)$.  Then, there exists unique $e(G,P)$ and $e(G)$ in $[0,\infty)$ such that \begin{align}
\label{certaintyRnonra}
E_{P} U(e(G,P))=U(e(G,P))  = &  E_{P} U(G(\cdot)), \, \forall P \in \mathcal{Q}^{T} \\
\label{certaintyrobRnonra}
\inf_{P \in \mathcal{Q}^{T}} E_{P} U(e(G))=U({e}(G)) =&  \inf_{P \in \mathcal{Q}^{T}} E_{P} U(G(\cdot)).
\end{align}
Moreover, $e(G,P) \leq E_{P} G(\cdot)$ for all $P \in \mathcal{Q}^{T}$ and $${e}(G)=\inf_{P \in \mathcal{Q}^{T}} e(G,P) \leq \inf_{P \in \mathcal{Q}^{T}} E_{P} G(\cdot).$$
Furthermore the multiple-priors  risk premium defined by    ${\lambda}(G):= \sup_{P \in \mathcal{Q}^{T}} \lambda(G,P)$ satisfies  $$0 \leq {\lambda}(G) \leq \sup_{P \in \mathcal{Q}^{T}}  E_{P}G(\cdot) -{e}(G).$$
\end{proposition}
Note that \eqref{certaintyRnonra} is true if we assume only that $E_{P} U^{-}(G(\cdot))<\infty$ for all $P \in \mathcal{Q}^{T}$.\\
\begin{proof}
See Section \ref{appendix2}.
\end{proof}

\noindent Finally, we consider two investors $A$ and $B$ whose respective  utility functions $U_{A}$ and $U_{B}$ satisfy  Assumption \ref{utilitydefdif} and are concave. Recall that in the uni-prior case with  $ \mathcal{Q}^{T}=\{P\}$ investor $A$ has greater absolute risk-aversion than investor B (i.e. $r_{A}(x) \geq r_{B}(x)$  for all $x>0$) if and only if investor $A$ is globally more risk averse than investor $B$, in the sense that the certainty equivalent of every contingent claim is smaller for $A$ than for $B$ (i.e. $e_A(G,P) \leq e_B(G,P)$ for any $G \in \mathcal{W}_{T}^{0,+}$) see \citep{Pr65}. We propose the following  generalization of this result in the multiple-priors  framework.
\begin{proposition}
\label{rrandce}
Let $U_{A}$, $U_{B}$ be {concave}  utility functions satisfying Assumption \ref{utilitydefdif}. Let $\mathcal{W}^{+}_{T}(U_{A,B}):=\mathcal{W}^{+}_{T}(U_{A}) \cap \mathcal{W}^{+}_{T}(U_{B})$ (see \eqref{WTU}). \\
\noindent $1.$ If for all $x > 0$, $r_{A}(x) \geq r_{B}(x)$ then $e_{A}(G) \leq {e}_{B}(G)$ for all $G  \in \mathcal{W}^{+}_{T}(U_{A,B})$.\\
\noindent $2.$  If for  all $G \in \mathcal{W}^{+}_{T}(U_{A,B})$,  $e_{A}(G) < {e}_{B}(G)$ then $r_{A}(x) \geq r_{B}(x)$ for all $x > 0$.
\end{proposition}
\begin{proof}
See Section \ref{appendix2}.
\end{proof}\\

Proposition \ref{rrandce} shows that the absolute risk aversion  allows the ranking of the multiple-priors  certainty equivalent despite  the presence of uncertainty (and thus uncertainty aversion). The reason for this is related  to the multiple-priors  representation we have chosen and the fact that the two agents have the same set of priors.

\begin{remark}
\label{aprm}
We briefly make the link with the monetary  risk measures introduced in \citep{Art99}, see also \citep{Car09}.
Let $\mathcal{X} \subset \mathcal{W}^{0}_{T}$ be  a linear space of random variables (containing the constant random variables).
A monetary risk measure is a mapping $\rho: G \in \mathcal{X} \mapsto \rho(G) \in \mathbb{R} \cup \{ \pm \infty\}$ that verifies the {monotonicity} and the {cash invariance} properties (see  \citep[Section 4.1]{fs}).
The measure  is said to be a normalized  if $\rho(0)=0$. 
One may measure the risk of a position using 
$\rho_{x}:\, G \in \mathcal{X} \mapsto p(-G,x)$ for some  $x \geq 0$ fixed
(see for example \citep[Definition 1.2]{Car09}) or
 consider
$\rho :\, G \in \mathcal{X} \mapsto \pi(-G).$ Indeed
under  Assumptions  \ref{Sass} and \ref{NAQT}, $\rho$ is a  normalized  convex monetary risk measure  on  $\mathcal{W}_{T}^{0}$.
If $U$ satisfies Assumption \ref{utilitydefdif} and under the assumptions of Proposition \ref{pivsp},  
$\rho_{x}$ is a monetary risk measure on  $\mathcal{W}_{T}^{0,bo}$. If one also assumes that
  $u(0,x)>-\infty$, then   $\rho_{x}$ is a  convex monetary risk  measure  on  $\{G \in \mathcal{W}_{T}^{0,bo}, \, u(-G,z)<\infty, \forall z \in \mathbb{R}\}$. If furthermore   $u(0,x-\delta)<u(0,x)$ for all $\delta>0$, then $\rho_{x}$ is normalized.

The proof is classical (see for instance \citep[Proposition 1.11]{Car09})   and   therefore omitted.

\end{remark}

\section{Example}
\label{example}
To illustrate the previous results we provide the following one period non-dominated example. 
Let  $\Omega=  [-1,\infty)$. We assume that the risky asset is given by $S_{0}=1$ and  $S_{1}(\o)=1+\o$. 
Here $\Phi=\mathbb{R}$ and Assumption \ref{Sass} is satisfied.  
We will compute the utility indifference price and the superreplication price of  the contingent claim $G(\o)=1_{[1,\infty)}(\o)$. 

Fix some  $p_{1},p_{-1} \in (0,1)$ such that $0<p_{1}+p_{-1}<1$. For $c \geq 1$, let $P_{c}$ be the probability measure on $\Omega$ be defined by 
$P_c(\cdot)=p_{1} \delta_{c}(\cdot)+p_{-1}\delta_{-1}(\cdot)+ (1-p_1-p_{-1}) \delta_{0}(\cdot)$ where $\delta_x(B):=1 \Leftrightarrow x \in B$ for all $x\in \Omega$ and $B\in \Bc(\Omega)$.  We assume furthermore that   $p_{1} < p_{-1}$: The agents think that the risky asset is more likely to go down than up. The set of priors is given by  $\mathcal{Q}:=\mbox{Conv}\{P_{c}, c \geq 1\}.$ 
Using \citep[Corollary 7.21.1 p130]{BS}, $T:c \in \mathbb{R} \mapsto P_{c}(\cdot) \in \mathcal{P}(\O)$ is a homeomorphism, so $\{P_{c}, c \geq 1\}=T([1,\infty))$ is a Borel set of  $\mathcal{P}(\O)$. As in  \citep[Proofs for Section 2.3]{Bart16}, we find that $\mathcal{Q}$ is a Borel set. 
So $\mbox{graph}(\mathcal{Q})=\O \times \mathcal{Q}$ is also a Borel set and  a fortiori an analytic set in $\O \times  \mathcal{P}(\O)$. Assumption \ref{Qanalytic}  is verified.  It is  clear that $NA(P)$ holds true for any $P \in \mathcal{Q}$. So  $NA(\mathcal{Q})$ also holds true.  Moreover the set $\mathcal{Q}$ is not dominated. Indeed for a dominating measure $P$, $P(\{c\})>0$ for all $ c \geq 1$ which is impossible.  \\
Finally, we take  a sequence of utility functions  defined by $U_{n}(x)=-e^{-r_{n}x}$  for $x\geq 0$ and $U_{n}(x)=-\infty$ otherwise where $\lim_{n \to +\infty} r_{n}=+\infty$. The assumptions of Theorem \ref{t2bis} item 1 are satisfied for all $x_0>0$, see Remark \ref{rembound1}. Note that as $\lim_{n \to +\infty} r_{n}=+\infty$ we will assume  that $r_n>\mbox{ln} \left(\frac{p_{-1}}{p_{1}}\right)$. 
\\
First, we focus on the admissibility condition. For  $x \geq 0$,  we get that 
\begin{align*}
 \mathcal{A}(0,x) & =  \left\{h \in \mathbb{R},\; -\frac{x}{c} \leq h\leq x, \; \forall \, c \geq 1\right\}=\left\{h \in \mathbb{R},\; 0 \leq h \leq x\right\} \\
\mathcal{A}(G,x)&=\left\{h \in \mathbb{R},\; \frac{1-x}{c} \leq h \leq x, \; \forall \, c\geq 1\right\}= \{h \in \mathbb{R},\;  \max(0,1-x) \leq h \leq x \}.
\end{align*}
It is easy to see  that $\pi(G)=\frac{1}{2}$. 
Note that both $\mathcal{A}(0,x)$ and $\mathcal{A}(G,x)$ are included in the half-real line: The agent cannot sell the risky asset in \eqref{theeq}. This implies that  the worst case prior in $\mathcal{Q}$  corresponds to  $P_1$ (recall that $c\geq 1$).\\
We   evaluate first $u_{n}(0,x)$ for $x \geq 0$. 
\begin{align}
\nonumber u_{n}(0,x)
&=\sup_{0 \leq h \leq x} -e^{-r_{n}x} \left(p_{1} e^{-r_{n} h} + 1-p_{1}-p_{-1}+ p_{-1} e^{-r_{n}(-h)}\right)
=-e^{-r_{n}x}.
\end{align}
Indeed, the minimum of $h \mapsto p_{1} e^{-r_{n} h} + 1-p_{1}-p_{-1}+ p_{-1} e^{r_{n}h}$ is reached for $\frac{1}{2r_{n}} \mbox{ln} \left(\frac{p_{1}}{p_{-1}}\right)<0$  since  $p_{1} < p_{-1}$.\\
If  $y<\frac{1}{2}=\pi(G),$ then $\mathcal{A}(G,y) \neq \emptyset$ and $u_{n}(G,y)=-\infty$. Let $y \geq \frac{1}{2}$.
 \begin{align*}
u_{n}(G,y)
&=\sup_{\max(0,{1-y}) \leq h\leq y} -e^{-r_{n}y} \left(p_{1} e^{-r_{n} (h-1)} + 1-p_{1}-p_{-1}+ p_{-1} e^{-r_{n}(-h)}\right).
\end{align*}  
Remark that the minimum of $h \mapsto p_{1} e^{-r_{n} (h-1)} + 1-p_{1}-p_{-1}+ p_{-1} e^{r_{n}h}$ is reached for $h^*_{n}:=\frac{1}{2}- \frac{1}{2r_{n}} \mbox{ln} \left(\frac{p_{-1}}{p_{1}}\right)$. As we have assumed that $p_{-1}>p_1$, $r_n>\mbox{ln} \left(\frac{p_{-1}}{p_{1}}\right)$ for all $n$,   $0<h^*_{n}<\frac{1}{2}$ and  
we obtain that \begin{align}
\label{ung1}
 u_{n}(G,y)=
\begin{cases} 
-e^{-r_{n}y} \left( 2e^{\frac{r_{n}}{2}}\sqrt{p_{1}p_{-1}}+(1-p_{1}-p_{-1}) \right) \;  \mbox{if $y \geq \frac{1}{2}+ \frac{1}{2r_{n}} \mbox{ln} \left(\frac{p_{-1}}{p_{1}}\right)$}\\
-e^{-{r_{n}y}} \left( p_{1}e^{r_{n}y}+p_{-1}e^{r_{n}(1-y)}+(1-p_{1}-p_{-1})\right) \; \mbox{otherwise}.
\end{cases}
\end{align}
Now starting from a wealth $x  \geq 0$, we want to find the smallest $z \in \mathbb{R}$ such that  $u_n(G,x+z) =u_{n}(0,x)$. 
Fix $x >0$. As $\lim_{n \to +\infty} u_{n}(0,x)=0$ there is some $N_{x}$ such that for $n \geq N_{x},$  $u_{n}(0,x) \geq -\frac{p_{1}}{2}$. 
If $\frac{1}{2} \leq x+z < \frac{1}{2}+\frac{1}{2r_{n}} \mbox{ln} \left(\frac{p_{-1}}{p_{1}}\right),$
 for $n \geq N_{x},$  \eqref{ung1} shows  that 
 $$u_{n}(G,x+z)\leq -p_{1}< u_{n}(0,x).$$ Thus for $n \geq N_{x},$ we need to assume that 
  $x+z \geq  \frac{1}{2}+ \frac{1}{2r_{n}} \mbox{ln} \left(\frac{p_{-1}}{p_{1}}\right)$. 
 Then \eqref{ung1} implies that
\begin{align}
\label{piex}
p_{n}(G,x)= \frac{1}{2}+ \frac{1}{r_{n}} \mbox{ln}\left( 2\sqrt{p_{1}p_{-1}}+e^{-\frac{r_{n}}{2}}(1-p^{}_{1}-p^{}_{-1})\right) \leq \pi(G) 
\end{align}
as soon as $x+ \frac{1}{r_{n}} \mbox{ln}\left( 2 p_1+e^{-\frac{r_{n}}{2}}\sqrt{\frac{p_{1}}{p_{-1}}}(1-p^{}_{1}-p^{}_{-1})\right) \geq 0$. 
{ In this case $p_{n}(G,x)$ is well-defined for $x<\pi(G).$ For example, if 
$p_{-1}= 2/3$, $p_{1}=1/4$ and $r_n=n$, this last inequality is always satisfied for $x=0.1$ as soon as $n\geq 7$ and then  $p_{7}(G,0.1)=0,471$ 
or for $x=0.4$ as soon as $n\geq 1$ and then $p_{1}(G,0.4)=0,417.$}

{
\begin{remark}
\label{remrate}
Note that the convergence rate, even in the presence of uncertainty, is only driven by the risk aversion $r_{n}$ (see \eqref{piex}).   
The optimal strategy for $u_{n}(G,y)=u^{P_{1}}_{n}(G,y)$ (recall that $P_1$ is the worst-case scenario) is given, for $n$ large enough, by $ \frac{1}{2}-\frac{1}{2r_{n}} \mbox{ln} \left(\frac{p_{-1}}{p_{1}}\right)$ which is different from the superreplication strategy $h=\frac12$ but converges to it with the risk aversion coefficient. Proving such result in a general setting i.e. finding a kind of accumulation point of the optimal  strategies's sequence which is a superhedging  price, is left for further research.
\end{remark}
}
\section{Proofs of the results}
\label{secproof}
We  borrow some ideas from \citep{CR06} adapted to the multiple-priors set-up.
\subsection{Proof of Theorems  \ref{t2} and   \ref{t2bis}}
\label{proofsec}
 The proof is based on two ingredients: Some closure property (see Lemma \ref{L3}) and some boundness property (see Lemma \ref{L10} which has already been used in Proposition \ref{pivsp} to fix some integrability issues).  We first  introduce, the set of terminal wealth including the possibility of throwing away money starting from capital $x \in \mathbb{R}$
\begin{align*}
\mathcal{C}_{x}^{T}:=\{V_{T}^{x,\phi},\; \phi \in \Phi\} - \mathcal{W}^{0,+}_{T}.
\end{align*}
In the sequel we will write $X \in \mathcal{C}_{x}^{T}$ if there exists some $\phi \in \Phi$ and $Z \in \mathcal{W}^{0,+}_{T}$ such that $X=V_{T}^{x,\phi} - Z$ $\mathcal{Q}^{T}$-q.s. Under the assumptions of Lemma \ref{L3} the set $\mathcal{C}_{x}^{T}$ has a classical closure property (in the $\mathcal{Q}^{T}$ quasi-sure sense, see \citep[Theorem 2.2]{BN}). Note that the same comment as in  Remark \ref{Theo22} applies.
\begin{lemma}
\label{L3}
Assume that Assumptions \ref{Sass} and    \ref{NAQT} hold true.
Fix some $z \in \mathbb{R}$ and let $B \in \mathcal{W}_{T}^{0}$ such that $B  \notin \mathcal{C}^{T}_{z}$.  Then there exists some $\varepsilon>0$
such that
\begin{align}
\label{Eqvare}
\inf_{\phi \in \Phi} \sup_{P \in \mathcal{Q}^{T}} P(V_{T}^{z,\phi} < B -\varepsilon)> \varepsilon.
\end{align}
\end{lemma}
\begin{proof}
Assume that \eqref{Eqvare} does not hold true.  Then, for all $n \geq 1$, there exists some $\phi_{n} \in \Phi$ such that  $P(V_{n} < B -\frac{1}{n})  \leq \frac{1}{n}$ for all $P \in \mathcal{Q}^{T}$, where $V_{n}:=V_{T}^{z,\phi_{n}}$. Set 
$$K_{n}:= \left(V_{n} - \left(B- \frac{1}{n}\right) \right) 1_{\{V_{n} \geq B -\frac{1}{n}\}} \in \mathcal{W}_{T}^{0,+}.$$ 
Then $V_{n} - K_{n} \in \mathcal{C}^{T}_{z}$. Moreover, for all  $P \in \mathcal{Q}^{T},$  
$P( |V_{n} -K_{n} -B|> \frac{1}{n}) = P\left(V_{n} < B -\frac{1}{n}\right) \leq  \frac{1}{n}.$ Thus $\lim_{n \to +\infty} \sup_{P \in \mathcal{Q}^{T}}  P( |V_{n} -K_{n} -B|> \frac{1}{n}) =0 $.  If we prove that
there exists a subsequence $(n_{k})_{k \geq 1}$ such that $(V_{n_{k}} - K_{n_{k}})_{k \geq 1}$ converges to $B$ $\mathcal{Q}^{T}$-q.s. (i.e. on a $\mathcal{Q}^{T}$-full measure set), \citep[Theorem 2.2]{BN} implies that    $B \in \mathcal{C}^{T}_{z}$. This contradiction will achieve the proof.
So fix $\eta>0$ and consider the  sub-sequence $(V_{n_{k}} - K_{n_{k}})_{k \geq 1}$ such that
$$\sup_{P \in \mathcal{Q}^{T}} P(A_{k}) \leq \frac{1}{2^{k}} \mbox{ 
where } A_{k}:=\left\{ |V_{n_{k}} - K_{n_{k}}-B| > \frac{1}{n_k}\right\}.$$ 
As $\sum_{k \geq 1} \sup_{P \in \mathcal{Q}^{T}} P(A_{k})
<\infty,$ 
 Borel-Cantelli's Lemma for capacity (see \citep[Lemma 5]{DHP11}) 
implies that $\sup_{P \in \mathcal{Q}^{T}} P(\limsup_{k} A_{k})=0$.
Hence  $\O^{T} \backslash{\limsup_{k} A_{k}}$ is a $\mathcal{Q}^{T}$-full measure set
on which   $|V_{n_{k}}(\cdot) - K_{n_{k}}(\cdot)-B(\cdot)| \leq \eta$ holds true for $k$ big enough.
\end{proof}\\

Under suitable assumptions, the next proposition  establishes that  whatever the strategy is, the wealth at time $T$ starting from capital $x$ is uniformly bounded.
\begin{lemma}
\label{L10}
 Assume that Assumptions \ref{Qanalytic}, \ref{Sass}, \ref{NAQT} and \ref{SalphaI}  hold true.
 Then, for all $x\in \mathbb{R}$, $\phi \in \mathcal{A}(0,x)$ and  $ 0 \leq t \leq T$, \begin{align}
 \label{Mt}
 |V_{t}^{x,\phi}(\cdot)| \leq  |x| M_{t}(\cdot) \; \mathcal{Q}^{t}\mbox{-q.s}
 \end{align}
   where $M_0:=1$ and $M_{t}(\o^{t}):=\prod_{s=1}^{t}\left(1+ \frac{|\Delta S_{s}(\o^{s})|}{\alpha_{s-1}(\o^{s-1})}\right)$. Moreover, $M_{t}$ and $V_{t}^{x,\phi}$ belong to $\mathcal{W}^{r}_{t}$  for all $0 \leq t \leq T$ and
   $0 < r<\infty$.
\end{lemma}

\begin{proof}
We  use  similar arguments  as in the proof of   \citep[Theorem  3.6]{BC16}. 
Let $x \leq 0$ and $ \phi \in \mathcal{A}(0,x)$. Then $V_T^{0,\phi} \geq 0  \; \mathcal{Q}^{T}\mbox{-q.s}$ and by  NA($\mathcal{Q}^{T}$) and \citep[Lemma A.33]{BC16},  $V_t^{0,\phi} \geq 0  \; \mathcal{Q}^{t}\mbox{-q.s}$ and $V_t^{0,\phi} = 0  \; \mathcal{Q}^{t}\mbox{-q.s}$. So \eqref{Mt} holds trivially true. 
So fix  $x > 0$ and $ \phi \in \mathcal{A}(0,x)$.  For all   $ 1 \leq t \leq T$  and $\o^{t-1} \in \O_{NA}^{t-1}$ (recall Proposition \ref{alpha} for the definition of $\O_{NA}^{t-1}$), we denote by $\phi_{t}^{\perp}(\o^{t-1})$ the
orthogonal projection of $\phi_{t}(\o^{t-1})$ on the vector space ${D}^{t}(\o^{t-1})$ (see again Proposition \ref{alpha}). We have for all $\o^{t-1} \in \O^{t-1}_{NA}$ that
 \begin{align}
\label{phiorth}
\phi_{t}(\o^{t-1}) \Delta S_{t}(\o^{t-1},\cdot)=\phi^{\perp}_{t}(\o^{t-1}) \Delta S_{t}(\o^{t-1},\cdot) \; \mathcal{Q}_{t}(\o^{t-1})\mbox{-q.s.}
\end{align}
see \citep[Remark 3.10]{BC16}.
As  $V_{T}^{x,\phi} \geq 0$ $\mathcal{Q}^{T}$-q.s. and as Assumptions \ref{Qanalytic}, \ref{Sass} and \ref{NAQT} hold true,     \citep[Lemma A.33]{BC16}  applies together with \citep[Lemma 3.4]{Nutz} and  $\mathcal{H}^{t-1}:=\{\o^{t-1} \in \O^{t-1},\;  V_{t-1}^{x,\phi}(\o^{t-1}) +\phi_{t}(\o^{t-1}) \Delta S_{t}(\o^{t-1},\cdot) \geq 0 \; \mathcal{Q}_{t}(\o^{t-1}) \mbox{-q.s.}\} \in \Bc_c(\O^{t-1})$ is a $\mathcal{Q}^{t-1}$-full measure set.   Fix now some $1 \leq t \leq T$, $\o^{t-1} \in \mathcal{H}^{t-1}\cap \O^{t-1}_{NA}$. We prove   that \begin{align}
\label{phiperp}
|\phi^{\perp}_{t}(\o^{t-1})| \leq \frac{|V_{t-1}^{x,\phi}(\o^{t-1})|}{\alpha_{t-1}(\o^{t-1})}.
\end{align}
If $\phi^{\perp}_{t}(\o^{t-1})=0$ there is nothing to prove and one may assume  that $\phi^{\perp}_{t}(\o^{t-1}) \neq0$. First, using \eqref{phiorth} and $\o^{t-1} \in \mathcal{H}^{t-1} \cap \O^{t-1}_{NA}$, we get that  \begin{align}
\label{c1}
 V_{t-1}^{x,\phi}(\o^{t-1})+ \phi^{\perp}_{t}(\o^{t-1})\Delta S_{t}(\o^{t-1},\cdot) \geq 0 \; \mathcal{Q}_{t}(\o^{t-1}) \mbox{-q.s}.
 \end{align}
Now, we proceed by contradiction and assume that \eqref{phiperp} does not hold true. Let
$$B:= \{ \phi^{\perp}_{t}(\o^{t-1})\Delta S_{t}(\o^{t-1},\cdot) < -\alpha_{t-1}(\o^{t-1}) |\phi^{\perp}_{t}(\o^{t-1})|\}.$$ From Proposition \ref{alpha}, there exists some $P_{\phi} \in \mathcal{Q}_{t}(\o^{t-1})$ such that $P_{\phi}(B) > \alpha_{t-1}(\o^{t-1})>0$. But,  for all $\o_{t} \in B$
\begin{align*}
V_{t-1}^{x,\phi}(\o^{t-1})+  \phi^{\perp}_{t}(\o^{t-1}) \Delta S_{t}(\o^{t-1},\o_{t}) < |V_{t-1}^{x,\phi}(\o^{t-1})| - \alpha_{t-1}(\o^{t-1}) |\phi^{\perp}_{t}(\o^{t-1})| <0,
\end{align*}
a contradiction with \eqref{c1} and therefore \eqref{phiperp} holds true.\\
Now, we prove by induction that \eqref{Mt} holds true for all $0 \leq t \leq T$.
For $t=0$ this is trivial. Assume  that for some $t \geq 1$, there exists some $\mathcal{Q}^{t-1}$-full measure set  $\widetilde{\O}^{t-1} \in \mathcal{B}_{c}(\O^{t-1})$ on which  \eqref{Mt} is true at stage $t-1$. Let
 $$\O^{t}_{EQ}:=\{(\o^{t-1},\o_{t}) \in \O^{t-1}\times \O_t,\;\phi_{t}^{\perp}(\o^{t-1}) \Delta S_{t}(\o^{t-1},\o_{t})=\phi_{t}(\o^{t-1}) \Delta S_{t}(\o^{t-1},\o_{t})\} .$$
It is clear that $\O^{t}_{EQ} \in \mathcal{B}_{c}(\O^t)$. For some $P=P_{t-1} \otimes p_{t}  \in \mathcal{Q}^{t}$,  \eqref{phiorth} and   Fubini's Theorem (see  \citep[Proposition 7.45 p175]{BS}) imply that
$P(\O^{t}_{EQ})=1$ (recall that $\O^{t-1}_{NA}$ is of $\mathcal{Q}^{t-1}$ full measure). 
 Set $\widehat{\O}^{t-1}:=\widetilde{\O}^{t-1} \cap \mathcal{H}^{t-1} \cap \O^{t-1}_{NA}$ and $\widetilde{\O}^{t}= \Omega^{t}_{EQ} \cap \left(\widehat{\O}^{t-1} \times \O_{t} \right)$. It is clear that $\widetilde{\O}^{t} \in \Bc_c(\Omega^t)$ and is a $\mathcal{Q}^{t}$-full measure set. For all
 $\o^{t}=(\o^{t-1},\o_{t}) \in \widetilde{\O}^{t}$ 
   \begin{align*}
   |V_{t}^{x,\phi}(\o^{t-1},\o_{t})|= & |V_{t-1}^{x,\phi}(\o^{t-1})+ \phi^{\perp}_{t}(\o^{t-1}) \Delta S_{t}(\o^{t-1},\o_{t})|\\
				  \leq  &  |V_{t-1}^{x,\phi}(\o^{t-1})| \left(1 + \frac{|\Delta S_{t}(\o^{t-1},\o_{t})|}{\alpha_{t-1}(\o^{t-1})}\right)\\
				  \leq   & x M_{t-1}(\o^{t-1})  \left(1 + \frac{|\Delta S_{t}(\o^{t-1},\o_{t})|}{\alpha_{t-1}(\o^{t-1})}\right)
   \end{align*}
and  \eqref{Mt} is proved. 
For all $0 \leq r <\infty$ and $1 \leq s \leq T$,  $\Delta S_{s}, \;\frac{1}{\alpha_{s}} \in \mathcal{W}^{r}_{s}$  (see Assumption  \ref{SalphaI}), so  both $M_{t}$ and $V_{t}^{x,\phi}$ belong to $\mathcal{W}^{r}_{t}$  for all $1 \leq t \leq T$.
\end{proof}\\

We are now in position to prove Theorem \ref{t2} and \ref{t2bis}. \\
\begin{proof}[\textit{ of Theorems \ref{t2} and \ref{t2bis}}]
 Let $G \in  \mathcal{W}^{0,bo}_{T}$ and $b \geq 0$ such that $G \geq  -b$ $\mathcal{Q}^{T}$-q.s. and fix some  $x_0 >0.$  \\
We prove first Theorem \ref{t2} and the first item of Theorem \ref{t2bis}. 
As in \citep{CR06}, we may replace $U_{n}$ by $\hat{U}_{n}:=\alpha_{n} U_{n} + \beta_{n}$ for  some $\alpha_{n} >0$, $\beta_{n} \in \mathbb{R}.$ If $U_n$ is  concave, strictly increasing or twice continuously differentiable, $\hat{U}_{n}$ will display the same property. Thus under the assumptions of Theorem \ref{t2}, the $\hat{U}_{n}$ are concave and satisfy Assumption \ref{SalphaI}. 
Moreover, the absolute risk aversion and the utility indifference price for $U_{n}$ and $\hat{U}_{n}$ (which will be denoted by $\hat{p}_{n}(G,x_0)$)  are the same (see \eqref{absorisk},  \eqref{theeq} and 
\eqref{prixut}).  Thus  it is enough to show that  $\lim_{n \to +\infty} \hat{p}_{n}(G,x_0)=\pi(G)$. As  $U'_{n}(x_0) \in (0,\infty)$ (recall that $U_{n}$ is {concave} and strictly increasing),
we choose  $\alpha_{n}=\frac{1}{U'_{n}(x_0)}$ and $\beta_{n}=-\frac{U_n(x_0)}{U'_{n}(x_0)}$  which leads to  $\hat{U}_{n}(x_0)=0$ and  $\hat{U}'_{n}(x_0)=1$ for all $n \geq 1$. Note that under the assumptions of Theorem \ref{t2bis} item 1,  the $\hat{U}_{n}$ are concave and there exist some $N>1$ and $k>0$ such that   $\hat{U}_{n}(x) \leq k$ for all $n \geq N$ and $x \in \mathbb{R}.$ 
We denote by $\hat{u}_{n}(G,x) $ the value function for $\hat U_n$.  
The choice of  $\alpha_{n}$ and $\beta_{n}$ implies that $0 \in {\cal A}(\hat{U}_n,0,x_0)$ and that\begin{align}
\label{enfin}
\hat{u}_{n}(0,x_0) \geq \inf_{P \in \mathcal{Q}^{T}} E_{P} \hat{U}_{n}(x_0) = 0.
\end{align}
We treat first the case  $\pi(G)=+\infty$. By definition  for all $z \in \mathbb{R}$, $n \geq 1$, $\emptyset=\mathcal{A}(G,z)=\mathcal{A}(\hat{U}_{n},G,z)$ and $\hat{u}_{n}(G,x_0+z)=-\infty$ (see  \eqref{theeq}).
 So, \eqref{prixut} and \eqref{enfin} show that $\hat{p}_{n}(G,x_0)=+\infty$ for all $n \geq 1$. The claim is proved.\\
Assume now that $\pi(G)<\infty$. 
Proposition \ref{inf} 
implies that $\hat{p}_{n}(G,x_0) \leq \pi(G)<\infty$ and $\lim_{n \to +\infty} \hat{p}_{n}(G,x_0)=\pi(G)$ will hold if $\liminf_{n} \hat{p}_{n}(G,x_0) \geq \pi(G)$. Assume that this is not the case. Hence there is subsequence $(n_{k})_{k \geq 1}$ and some $\eta >0$ such that
$\hat{p}_{n_{k}}(G,x_0) \leq  \pi(G)-\eta$ for all $k \geq 1$. Since $x_0>0$, we may and will assume that $ \eta < x_0$.
By definition of $\hat{p}_{n_{k}}(G,x_0)$ we have that $$\hat{u}_{n_{k}}(G,x_0 + \pi(G)-\eta) \geq \hat{u}_{n_{k}}(0,x_0).$$
If 
 $\lim_{k \to +\infty} \hat{u}_{n_{k}}(G,x_0 + \pi(G)-\eta) = -\infty$ is proved,  {$\liminf_{k \to +\infty} \hat u_{n_{k}}(0,x_0)=-\infty$} follows and  contradicts \eqref{enfin}.
So, it remains to prove  that $\lim_{k \to +\infty} \hat{u}_{n_{k}}(G,y) = -\infty$ with $y:=x_0 + \pi(G)-\eta \in (\pi(G),x_0 + \pi(G))$.\\
First we show that  $x_0+ G \notin \mathcal{C}^{T}_{y}$.  Indeed if this is not the case,  there exists some $X \in \mathcal{W}_{T}^{0,+}$ and $\phi \in \Phi$ such that $x_0+G=V_{T}^{y,\phi}-X$ $\mathcal{Q}^{T}$-q.s. Hence $G \leq V_{T}^{y-x_0,\phi}$ $\mathcal{Q}^{T}$-q.s. and $y-x_0 \geq \pi(G)$ follows: A contradiction.
Applying  Lemma \ref{L3},
we get some $\varepsilon>0$ such that $\inf_{\phi \in \Phi} \sup_{P \in \mathcal{Q}^{T}} P (A_{\phi}) > \varepsilon$, where $A_{\phi}:=\{V_{T}^{y,\phi} < x_0+ G-\varepsilon\}.$
Note that we can always assume that $x_0 > \varepsilon$.
Hence for all $\phi \in \Phi$, there exists some $P_{\varepsilon,\phi} \in \mathcal{Q}^{T}$ such that $P_{\varepsilon,\phi} (A_{\phi}) > \varepsilon$.\\
Proposition \ref{pivsp}
and Theorem \ref{BN2} imply that  for all $ n\geq 1$,  $\mathcal{A}(\hat U_{n},G,y)=\mathcal{A}(G,y) \neq \emptyset$ since $y > \pi(G)$. Choose some $\phi \in \mathcal{A}(G,y)$.  \\ 
We first postulate  the assumptions of Theorem \ref{t2}.  
Using the monotonicity of $\hat{U}_{n}$, recalling $G(\cdot) \geq  -b$ $\mathcal{Q}^{T}$-q.s.,  \eqref{unstar} and Lemma \ref{L10}, we get for all $n \geq 1$ that
\begin{align}
\nonumber
 E_{P_{\varepsilon,\phi}} 1_{\Omega^{T}\backslash{A_{\phi}}} \hat U_{n}(V_{T}^{y,\phi}(\cdot)- G(\cdot)) &\leq  E_{P_{\varepsilon,\phi}} \hat U^{+}_{n}(V_{T}^{y+b,\phi}(\cdot))
 \leq  \hat U_n^{+}(x_0) +\sup_{P \in \mathcal{Q}^{T}} E_{P} \left(\left|V_{T}^{y+b,\phi}(\cdot)\right|\right) \hat U_n'(x_0)\\
 \nonumber
& \leq   (|y|+b)\sup_{P \in \mathcal{Q}^{T}} E_{P} \left(M_T(\cdot)\right)\\
\label{Step1}
& \leq   (x_0+ |\pi(G)|+b)||M_{T}||_1=:K<\infty.
\end{align}
Under the assumptions of Theorem \ref{t2bis} item 1, the last boundness  from above property in \eqref{Step1}   is still valid for  $K=k$ and  $n\geq N$. 
Now, as $\hat{U}_{n}(x_0-\varepsilon) \leq \hat{U}_{n}(x_0)=0$ we get that
\begin{align}
\label{Step2}
E_{P_{\varepsilon,\phi}} 1_{A_{\phi}}\hat{U}_{n} \left(V_{T}^{y,\phi}(\cdot)-G(\cdot)\right)  \leq \hat{U}_{n}(x_0-\varepsilon) E_{P_{\varepsilon,\phi}} 1_{A_{\phi}} \leq \varepsilon \hat{U}_{n}(x_0-\varepsilon).
\end{align}
So  \eqref{Step1} and \eqref{Step2} imply that    for all $n \geq N$ 
$$ \inf_{P \in \mathcal{Q}^{T}} E_{P}\hat{U}_{n} \left(V_{T}^{y,\phi}(\cdot)-G(\cdot)\right) \leq E_{P_{\varepsilon,\phi}}\hat{U}_{n} \left(V_{T}^{y,\phi}(\cdot)-G(\cdot)\right) \leq  K +\varepsilon \hat{U}_{n}(x_0-\varepsilon).$$
As this is true for all $\phi \in \mathcal{A}(G,y)=\mathcal{A}(\hat{U}_{n},G,y)$,  $\hat{u}_{n}(y,G) \leq  K +\varepsilon \hat{U}_{n}(x_0-\varepsilon)$ follows  for all $n\geq N$. Finally,
  \citep[Lemma 4]{CR06} (which uses the concavity of $\hat{U}_{n}$)  implies that  $\lim_{n \to +\infty} \hat{U}_{n}(x_0-\varepsilon)=-\infty$ and thus $\lim_{n \to +\infty} \hat{u}_{n}(G,y)=-\infty$ as claimed.\\
{We now prove the second item of Theorem \ref{t2bis}. The proof is similar to the one of  the first item of Theorem \ref{t2bis} and we only enlighten the main changes.  First, we do not modify the functions $U_n$. Let $k_1 >0$ be such that $U_n(x) \leq k_1$ for all $n \geq N$ and $x \in \mathbb{R}$. 
As $0 \in {\cal A}(U_n,0,x_0)$ for all $n \geq N$ 
\begin{eqnarray*}
u_{n}(0,x_0) \geq  U_n(x_0)>\inf_{n\geq N}U_n(x_0)  >-\infty,
\end{eqnarray*}
which is the pendant of \eqref{enfin}. 
 If $\pi(G)=+\infty$ the same arguments as above apply. Assume that $\pi(G)<+\infty$: Proposition \ref{inf} still applies and $\hat{p}_{n}(G,x_0) \leq \pi(G)$.  Fix $\phi \in \mathcal{A}(G,y)$ and let $y$, $A_{\phi}$ and $P_{\varepsilon,\phi}$ be as before: We prove directly that $\lim_{n \to +\infty} u_{n}(G,y)=-\infty$. 
Fix some $J>0$ and $C_J := \frac{1}{\e}\left(J+k_1\right)$. 
As $\lim_{n\to + \infty}U_n(x_0 -\e)=-\infty$, there exists 
$N_{J}\geq N$ such that for all $n \geq N_{J}$, $U_n(x_0 -\e)\leq -C_J$. 
\begin{align*}
 E_{P_{\varepsilon,\phi}} U_{n} \left(V_{T}^{y,\phi}(\cdot)-G(\cdot)\right) & \leq  
 E_{P_{\varepsilon,\phi}} 1_{\Omega^{T}\backslash{A_{\phi}}} U_{n}(V_{T}^{y,\phi}(\cdot)- G(\cdot)) + 
 E_{P_{\varepsilon,\phi}} 1_{A_{\phi}}U_{n} \left(V_{T}^{y,\phi}(\cdot)-G(\cdot)\right) \\
& \leq k_1 +U_n(x_0 -\e) P_{\varepsilon,\phi}(A_{\phi}) \leq k_1 + U_n(x_0 -\e)\varepsilon \leq -J.
\end{align*}
Thus  as  $N_{J}$ does not depend on $\phi$,  we obtain that for all  $n \geq N_{J}$, $u_{n}(y,G) \leq -J$.}
\end{proof}\\

\subsection{Proofs of Propositions \ref{certainequivdef} and \ref{rrandce}}
\label{appendix2}

\begin{proof}[\textit{of Proposition \ref{certainequivdef}}]
 Fix some
$P \in \mathcal{Q}^{T}$ and $G \in \mathcal{W}_{T}^{+}(U)$.  
As $E_{P} U^{-}(G(\cdot))<+\infty$, $G \in \mathcal{W}_{T}^{0,+}$ and $U$ is non-decreasing,
$E_{P} U(G(\cdot)) -U(0) \geq 0$ (note that $U(0)$ may be equal to $-\infty$). 
As  $P(G(\cdot)<\infty)=1$ and $U$ is  strictly increasing $U(G(\cdot))< \lim_{y \to +\infty} U(y)$ $P$-a.s. Together with  $E_{P} U^{+}(G(\cdot))<+\infty$, one concludes that
$E_P U(G(\cdot))-  U(y)<0$ for $y$ large enough and the intermediate value theorem implies that \eqref{certaintyRnonra}  holds true.\\
Now  $E_{P} U(G(\cdot)) \geq U(0)$ for all $P \in \mathcal{Q}^{T}$ implies that
$\inf_{P \in \mathcal{Q}^{T}}E_{P} U(G(\cdot)) -U(0) \geq 0$ (recall that $\sup_{P \in \mathcal{Q}^{T}}E_{P} U^-(G(\cdot))<\infty$). Moreover for some $P \in \mathcal{Q}^{T}$ as
$\inf_{P \in \mathcal{Q}^{T}}E_P U(G(\cdot))-  U(y) \leq E_P U(G(\cdot))-  U(y)<0$ for $y$ large enough, the intermediate value theorem implies  \eqref{certaintyrobRnonra}. \\
Now for any $Q \in \mathcal{Q}^{T}$, \eqref{certaintyRnonra}  implies that
\begin{align*}
 U(\inf_{P \in \mathcal{Q}^{T}} e(G,P)) \leq  U(e(G,Q))= E_{Q} U(G(\cdot)).
\end{align*}
This is true for any $Q \in \mathcal{Q}^{T}$, so  \eqref{certaintyrobRnonra}  implies that
\begin{align*}
 U(\inf_{P \in \mathcal{Q}^{T}} e(G,P)) \leq \inf_{P \in \mathcal{Q}^{T}} E_{P} U(G(\cdot)) =U({e}(G))
\end{align*}
and  $\inf_{P \in \mathcal{Q}^{T}} e(G,P) \leq {e}(G)$ follows from strict monotonicity of $U$. 
Now for any $Q \in \mathcal{Q}^{T}$, \eqref{certaintyRnonra}, \eqref{certaintyrobRnonra} and Jensen's inequality imply that
\begin{align*}
U({e}(G))= \inf_{P \in \mathcal{Q}^{T}} E_{P} U(G(\cdot))   \leq E_{Q} U( G(\cdot))
 =U(e(G,Q))\leq U \left(E_{Q} G(\cdot)\right). \end{align*}
Thus, by strict monotonicity of $U$, ${e}(G) \leq e(G,Q) \leq E_{Q} G(\cdot)$ and since this is true for all $Q \in \mathcal{Q}^{T}$, we find that  $${e}(G) \leq \inf_{P \in \mathcal{Q}^{T}}e(G,P) \leq \inf_{P \in \mathcal{Q}^{T}} E_{P} G(\cdot).$$
\end{proof}\\

\begin{proof}[\textit{of Proposition \ref{rrandce}}]
We adapt the  proof  of \citep[Proposition 2.47]{fs} to the multiple-priors framework.\\
$1.$ We first show that if for all $x > 0$, $r_{A}(x) \geq r_{B}(x)$, then
 $e_{A}(G,P) \leq {e}_{B}(G,P)$
 for all $G  \in \mathcal{W}^{+}_{T}(U_{A,B})$ and 
 $P \in \mathcal{Q}^{T}$.  This will imply that $e_{A}(G) \leq {e}_{B}(G)$ using Proposition \ref{certainequivdef}. 
Fix some   $G  \in \mathcal{W}^{+}_{T}(U_{A,B})$  and  $P \in \mathcal{Q}^{T}$. Let $D:=U_{B}((0,\infty)) \subset (-\infty,\infty)$ and define $F: D \to \mathbb{R}$  by $F(y)= U_{A}\left(U^{-1}_{B}(y)\right)$. Then on $D$ \begin{align}
 \label{fsec}
F'(\cdot)=\frac{U'_{A}(U_{B}^{-1}(\cdot))}{U^{'}_{B}(U_{B}^{-1}(\cdot))} \;  \mbox{and} \; F''(\cdot)= \frac{ U'_{A}(U^{-1}_{B}(\cdot))}{\left(U'_{B}(U^{-1}_{B}(\cdot))\right)^{2}} \left( r_{B}(U_{B}^{-1}(\cdot))-r_{A}(U_{B}^{-1}(\cdot))\right).
 \end{align}
As $U^{-1}_{B}(\cdot)>0$ on $D$,  $F$ is  increasing and concave on $D$ and  $U_{A}(x)=F (U_{B}(x))$ for all $x >0$.  Now let  $\d:=U_{B}(0) \in [-\infty,\infty)$ be  the lower bound of  $D$. We distinguish between two cases. First if $\d>-\infty$, we  extend $F$ by continuity in $\d$, setting $F(\d)= U_{A}\left(U^{-1}_{B}(\d)\right)=U_{A}(0) \in [-\infty,\infty)$. It is clear that $F(\d) \leq F(y)$ for all $y \in [\d,+\infty)$, that $F$ is concave on $[\d,+\infty)$  and that  $U_{A}(x)=F (U_{B}(x))$ holds also true for all $x \geq 0$.   Now, using 
\eqref{certaintyRnonra} and  Jensen's inequality, we get that \begin{small} \begin{align}
\label{UAFUB}
U_{A}({e}_{A}(G,P))=  E_{P} U_{A} (G(\cdot))= E_{P} F\left(U_{B}(G(\cdot))\right)
 &\leq F\left( E_{P} \left(U_{B}(G(\cdot))\right)\right)
 =F(U_{B}({e}_{B}(G,P)))=U_{A}({e}_{B}(G,P)).
 \end{align}\end{small}
Since $U_{A}$ is strictly increasing, we obtain that ${e}_{A}(G,P) \leq {e}_{B}(G,P)$ as claimed.\\
Now we treat the case where $\d=-\infty$. First $P(G>0)=1$. Indeed if $P(G=0)>0$,   $E_{P} U^{-}_{B} (G(\cdot))= E_{P} U^{-}_{B}(G(\cdot))1_{\{G>0\}}(\cdot)+ U^{-}_B(0) P(G=0)=+\infty$, a case that we have excluded. Thus
$P(G>0)=1$. Moreover $e_A(G,P)$ and ${e}_{B}(G,P)$ are positive. Else $U_{A}({e}_{A}(G,P))=-\infty$ while 
$E_P U_A(G(\cdot)) \geq -E_P U^-_A(G(\cdot))>-\infty$. Thus
the previous arguments apply and  we also obtain ${e}_{A}(G,P) \leq {e}_{B}(G,P)$.  \\
\noindent $2.$ Assume that $e_{A}(G) < e_{B}(G)$ for all $G \in \mathcal{W}_{T}^{+}(U_{A,B})$ and there exists some $x_0>0$ such that $r_{A}(x_0) < r_{B}(x_0)$. By continuity, there exists $\alpha>0$, such that $r_{A}(x) < r_{B}(x)$ on $(x_0-\alpha,x_0+\alpha)$. We can choose $\alpha$ such that $x_0-\alpha>0$. Let $I:=\left(U_{B}(x_0-\alpha), U_{B}(x_0+\alpha)\right) \subset D$, then $F$ is strictly convex on $I$ (see \eqref{fsec}). Fix $\widetilde{G} \in\mathcal{W}^{+}_{T}(U_{A,B})$   and set $G:= x_0-\alpha + 2 \alpha \frac{\widetilde{G}}{\widetilde{G}+1} \in  \mathcal{W}^{+}_{T}(U_{A,B})$. It is clear that  $G(\cdot)   \in (x_0-\alpha,x_0+\alpha)$. 
As in \eqref{UAFUB}, using Jensen inequality, the fact that  $F$ is (strictly) convex on $I$ we get that for any $P\in
\mathcal{Q}^{T}$
\begin{align}
\label{starUU}
U_{A}({e}_{A}(G,P)) =E_{P} F\left(U_{B}(G(\cdot)\right)) \geq F(E_{P} (U_{B}(G(\cdot)))=U_{A}(e_{B}(G,P)).
\end{align}
This implies that $e_{A}(G,P) \geq e_{B}(G,P)$ for all $P \in \mathcal{Q}^{T}$, thus $e_{A}(G)\geq e_{B}(G)$: A contradiction. Note that if $P$ is such that one can find some  $\widetilde{G}$ which is not constant then the  inequality in \eqref{starUU} is strict and one gets that $e_{A}(G,P)>e_{B}(G,P)$.
\end{proof}
\section{Extension to  random utility functions}
\label{apen}
Random utility functions capture very general situations where the preferences of the agent depend not only on her wealth but also on the path. At the starting date,  the agent might not know  exactly how her utility function will depend on her wealth. Moreover the shape of her utility function vary with the context  and can be updated as information is released. For instance, she could become more risk averse if the market exhibits a tendency to move lower and could take more risk in the opposite situation. Such behaviors are often  observed in financial markets.  An example of state-dependent utility function is the forward investment performance process introduced in \citep{Mu05}, see also \citep{KOZ17} for an extension to the multiple-priors framework.
\begin{definition}
\label{utilitydefdifR}
A random utility function  $U:\Omega^T \times (0,\infty) \rightarrow \mathbb{R}\cup \{- \infty\}$  satisfies the following  conditions
	\begin{itemize}
		\item[i)]
		for every $x>0$, $U \left(\cdot,x\right):~\Omega^{T}\rightarrow\mathbb{R}$ is universally-measurable,
		\item[ii)]
		{for all $\omega^{T} \in \Omega^{T}$,  $U \left(\omega^{T},\cdot\right):~(0,\infty)  \rightarrow\mathbb{R}$ is  non decreasing on $(0,\infty)$ and such $U \left(\omega^{T},x_{\o}\right)>-\infty$ for some $x_{\o}>0$.}
\end{itemize}
We extend $U$ by (right) continuity in $0$ and set  $U(\cdot,x)=-\infty$ if $x<0$.
\end{definition}
The next example  exhibits random utility functions  such that in addition 
$U_n \left(\omega^{T},\cdot\right)$  is  concave, strictly increasing and twice continuously differentiable on $(0,\infty)$.
\begin{example}
\label{randomuex}
Assume that the agent analyzes her gain or loss with respect to a (random) reference point $B$ rather than with respect to zero as suggested for instance by  \citep{kt}. Let $\overline{U}$ be a non-random {concave} function satisfying Assumption  \ref{utilitydefdif} and $B \in \mathcal{W}_{T}^{\infty,+}$ and set for all $\o^T \in \O^{T}$, $x \geq 0$, $U(\o^{T},x)=\overline{U}(x+||B||_{\infty}-B(\o^{T}))$ and $U(\o^{T},x)=-\infty$ for $x<0$.  \\
The second example  proposes to consider random absolute risk aversion.
The idea is to use classical utility functions but with random coefficients. For example, we may consider $U(\o^{T},x)= x^{\beta_{1}(\o^{T})}$ or $U(\o^{T},x)= -e^{-\beta_{2}(\o^{T})x}$ for $x \geq 0$ (and $U(\cdot,x)=-\infty$ for $x<0$) where $\beta_{1},\beta_{2} \in  \mathcal{W}_{T}^{0}$ and $0<\beta_{1}(\cdot) <1$, $\beta_{2}(\cdot) >0$ $\mathcal{Q}^{T}$-q.s.  We can imagine various situations for $\beta_2$  (which can be easily adapted for $\beta_{1}$):  The law of  $\beta_{2}$ under $P$ can be uniformly distributed on  $[\beta^P_{min},\beta^P_{max}]$ for all $P \in \mathcal{Q}^{T}$ (with $\beta^P_{max} \geq \beta^P_{min} >0$), alternatively it could follow a Poisson law  of parameter $\l_P>0$ for all $P \in \mathcal{Q}^{T}$. It could also be a function of some market parameters to model situations where the agent updates her utility function depending on market conditions. \end{example}

Let $(U_n)_{n \geq 1}$ be a sequence of utility functions satisfying Definition \ref{utilitydefdifR}. 
The following definition is  \eqref{theeq}    
 adapted  for general random utility functions.
 \begin{eqnarray*}
u_{n}(G,x)&:=& \sup_{\phi \in \mathcal{A}(U_{n},G,x)} \inf_{P \in \mathcal{Q}^{T}} E_{P} U_{n} \left(\cdot, V_{T}^{x,\phi}(\cdot)-G(\cdot)\right)\\
 \end{eqnarray*}
The generalization of Theorem \ref{t2}  for random utility functions will be stated for some fixed $x_0>0$ and requires  some further assumptions. The first one replaces the convergence of the absolute risk aversion to infinity:   $U_{n}(\cdot,x)$ goes to $-\infty$ with respect to $\inf_{P \in \mathcal{Q}^{T}} P$ for all $0 < x <x_0$.  
This is explained  in Lemma \ref{L4}, where we also give alternative conditions to \eqref{toprove}. {It also requires some uniform boundness from below assumption in $x_0$}.
\begin{assumption}
\label{un}
We have that  {$\sup_{n} ||U^{-}_{n}(\cdot,x_0)||_{1}<\infty$} and that for all $0 < x <x_0$ and $M > 0$,
\begin{align}
\label{toprove}
\lim_{n \to +\infty} \inf_{P \in \mathcal{Q}^{T}} P\left(U_{n}(\cdot,x) \leq -M\right)=1.
\end{align}
\end{assumption}
The second assumption allows to fix integrability issues for unbounded from above utility functions and states that the $U_{n}$ are sufficiently measurable and regular. 
\begin{assumption}
\label{u1}
{For all $\omega^{T} \in \Omega^{T}$,  $U \left(\omega^{T},\cdot\right)$ is   concave  and  twice continuously differentiable on $(0,\infty)$.} There exist some $x_{1} \geq x_0$ and some $q>1$ such that 
$$\sup_{n} ||U^{+}_{n}(\cdot,x_1)||_{1}<\infty \mbox{ and }  \sup_{n} ||U'_{n}(\cdot,x_1)||_{q}<\infty.$$ 
\end{assumption}
\begin{theorem}
\label{t1}
Assume that Assumptions \ref{Qanalytic}, \ref{Sass} and \ref{NAQT} holds true. Let $(U_n)_{n \geq 1}$ be a sequence of  random utility functions satisfying Definition \ref{utilitydefdifR} and let $G \in  \mathcal{W}^{0,bo}_{T}$. 
Assume that Assumption \ref{un} holds true for some $x_0>0$.  
{Suppose that either there exist some $N>1$ and $B  \in \mathcal{W}_{T}^{1,+}$ such that 
$U_n(\cdot,x) \leq B(\cdot)$ $\mathcal{Q}^{T}$-q.s. for all $x>0$ and $n \geq N$ or that  
 Assumptions \ref{SalphaI}  and \ref{u1} hold true}. 
Then
$\lim_{n \to +\infty} p_{n}(G,x_0)=\pi(G)$.
\end{theorem}
\begin{example}
\label{examplerandon}
{We give a concrete example for Theorem \ref{t1}.   For all $n \geq 1$, let $R_{n}$  be a random variable  uniformly distributed in $[b_{n},B_{n}]$ for all $P \in \mathcal{Q}^{T}$  with $b_{n}>0$, $\lim_{n \to +\infty} b_{n}=+\infty$. 
Set  for all $\omega^{T} \in \Omega^{T}$ $U_{n}(\o^{T},x)=-e^{-R_{n}(\o^{T})(x-1)}$ for $x \geq 0$ and $U_{n}(\o^{T},x)=-\infty$ for $x<0$.  
We choose $x_0=1$. Then  for all $M>0$ and $0<x<1$, $U_n(\cdot,x)\leq -M$ if and only if $R_n (\cdot)\geq \frac{\ln M}{1-x}$. As $\lim_{n \to +\infty} b_{n}=+\infty$, Assumption \ref{un}  is verified. }
\end{example}
\begin{proof}
We  outline briefly how the proof of Theorem \ref{t2} (and Theorem \ref{t2bis}) is modified.  The structure of the proof is similar to the one of item 2 of Theorem \ref{t2bis}, in particular we do not modify the functions $U_n$. \\
Under the assumptions of Theorem \ref{t1}, Proposition \ref{pivsp} (and thus Proposition \ref{inf}) is still valid replacing \eqref{Kx} by   \eqref{KxR} below.  \\
Assume first that Assumptions \ref{SalphaI}  and \ref{u1} hold true. Using \eqref{unstar} for $U_n(\o^{T},\cdot)$ (recall $x_1>0$), we get that  for all $x>0$, $\o^{T} \in \O^{T}$, for all  $P \in \mathcal{Q}^{T}$
\begin{align}
\nonumber
 E_{P} U_n^{+}(.,  V_{T}^{x,\phi}(.))  \leq & \sup_{P \in \mathcal{Q}^{T}} E_{P} U_n^{+}(\cdot,x_1) + \sup_{P \in \mathcal{Q}^{T}} E_{P} \left(\left|V_{T}^{x,\phi}(\cdot)\right|  U_n'(\cdot,x_1)\right)\\
 \nonumber
\leq & ||U_n^{+}(\cdot,x_1)||_{1}+ |x| \, ||M_{T}(\cdot) U_n'(\cdot,x_1)||_{1}\\
\nonumber
 \leq &  ||U_n^{+}(\cdot,x_1)||_{1}+ |x| \, ||M_{T}(\cdot)||_{p}  ||U_n'(\cdot,x_1)||_{q} \\
 \label{KxR}
 \leq & \sup_{n} ||U_{n}^{+}(\cdot,x_1)||_{1}+|x| \, ||M_{T}(\cdot)||_{p} \sup_{n}||U_{n}'(\cdot,x_1)||_{q}=:
 \overline{K}(x_1, x)<\infty,
\end{align}
where  Lemma \ref{L10}, $M_{T} \in \mathcal{W}^{p}_{T}$  (where $p$ verifies $\frac{1}{p}+\frac{1}{q}=1$),    Assumption \ref{u1} and \citep[Proposition 16]{DHP11} have been used. \\
{Now, if Assumptions \ref{SalphaI}  and \ref{u1} do not hold true but there exists some $B  \in \mathcal{W}_{T}^{1,+}$ such that 
$U_n(\cdot,x) \leq B(\cdot)$ for all $x>0$ and $n \geq N$, the last boundness property in \eqref{KxR} is still valid for  $\overline{K}(x_1, x)=||B(\cdot)||_{1}$.}

Remark now that $0 \in {\cal A}(U_n,0,x_0)$ for all $n \geq N$ (recall that $x_0>0$ and \eqref{KxR}). Now Assumption \ref{un} implies that for all $n \geq N$
\begin{eqnarray*}
u_{n}(0,x_0) \geq \inf_{P \in \mathcal{Q}^{T}} E_{P} U_{n}(\cdot,x_0) \geq -\sup_{n} ||U^{-}_{n}(\cdot,x_0)||_{1} >-\infty.
\end{eqnarray*}
Let $G \in  \mathcal{W}^{bo}_{T}$  and $b \geq 0$ such that $G \geq -b $ $\mathcal{Q}^{T}$-q.s. Let $\phi \in \mathcal{A}(G,y)$ and let $y$, $A_{\phi}$ and $P_{\varepsilon,\phi}$ be as in the proof of Theorems \ref{t2} and \ref{t2bis}.
Instead of  \eqref{Step1} (recall \eqref{KxR}), we use that 
\begin{align}
\label{Step1bis}
 E_{P_{\varepsilon,\phi}} 1_{\Omega^{T}\backslash{A_{\phi}}} U_{n}(\cdot,V_{T}^{y,\phi}(\cdot)- G(\cdot))  \leq \overline{K}(x_1,x_0+ |\pi(G)|+b).
\end{align}
The arguments to obtain \eqref{Step2} are more involved. First, fix some $J>0$ and set 
$$C_J := \frac{2}{\e}\left(J+\overline{K}(x_1,x_0)+ \overline{K}(x_1,x_0+ |\pi(G)|+b)\right) \mbox{ and } B_{J,n}:=\{U_{n}(\cdot,x_0-\varepsilon) \leq -C_{J}\}.$$  We  apply Assumption \ref{un}  (recall that $x_0 > \e$) and obtain some $N_{J}\geq {N}$ (which does not depend on $\phi$) such that for all $n \geq N_{J}$,
$$P_{\varepsilon,\phi}\left(B_{J,n}\right) \geq \inf_{P \in \mathcal{Q}^{T}} P\left(B_{J,n}\right) > 1-\frac{\varepsilon}{2}.$$
Then, for all $n \geq N_{J}$, $P_{\varepsilon,\phi}\left(B_{J,n}\cap A_{\phi}\right) >  \frac{\varepsilon}{2}$ (recall that $P_{\varepsilon,\phi} (A_{\phi}) > \varepsilon$) and  we get that
\begin{align*}
 E_{P_{\varepsilon,\phi}} 1_{A_{\phi}}U_{n} \left(\cdot,V_{T}^{y,\phi}(\cdot)-G(\cdot)\right) & \leq  E_{P_{\varepsilon,\phi}} 1_{A_{\phi} \cap B_{J,n}} U_{n} (\cdot,x_0-\varepsilon ) + E_{P_{\varepsilon,\phi}} 1_{A_{\phi} \backslash{B_{J,n}}}  U_{n} (\cdot,x_0 )\\
&\leq \frac{-\varepsilon C_{J} }{2} + \overline{K}(x_1,x_0)= -J - \overline{K}(x_1,x_0+ |\pi(G)|+b),
\end{align*}
using \eqref{KxR} and the definition of $C_{J}$.
 Combining the previous equation  with   \eqref{Step1bis},  we obtain that  for all $n \geq N_{J}$  
$$ \inf_{P \in \mathcal{Q}^{T}} E_{P}U_{n} \left(\cdot,V_{T}^{y,\phi}(\cdot)-G(\cdot)\right) \leq E_{P_{\varepsilon,\phi}}U_{n} \left(\cdot,V_{T}^{y,\phi}(\cdot)-G(\cdot)\right) \leq -J.$$
 As  $N_{J}$ does not depend on $\phi$,  we obtain that for all  $n \geq N_{J}$, $u_{n}(y,G) \leq -J$. Since  this is  true for all $J \geq 0$, $\lim_{n \to +\infty} u_{n}(G,y)=-\infty$ and the proof is complete.
\end{proof}\\

{We now make the link between \eqref{toprove} and the convergence of the absolute risk aversion. From now we take a sequence of 
utility functions  satisfying Assumption \ref{utilitydefdifR} and such that for all $\omega^{T} \in \Omega^{T}$,  $U_n \left(\omega^{T},\cdot\right)$ is   concave  and  twice continuously differentiable on $(0,\infty)$. The generalisation of the absolute risk aversion (see 
  \eqref{absorisk}) is}
$$
r_{n}(\o^{T},x):=-\frac{U_{n}^{''}(\o^{T},x)}{U_{n}^{'}(\o^{T},x)}.$$
 
\begin{lemma}
\label{L4}
Assume that  {$\sup_n \|U_{n}(\cdot,x_0)\|_1 <+ \infty$}  and that
there exists some $N\geq 1$, a  strictly positive random variable $\l$  and  some deterministic functions  $(\rho_{n})_{n \geq 1}$ such that for all  $n\geq N$, $U'_{n}(\cdot,x_0) \geq \l(\cdot)$, $r_{n}(\cdot,x) \geq \rho_{n}(x)$ and $\lim_{n \to +\infty} \rho_{n}(x)=+\infty$ for all   $ x \in (0,x_{0}]$.  Then  \eqref{toprove} holds true.
\end{lemma}
\begin{proof}
Suppose  for  all $\e>0$ such that $x_0>\e$ and all $C\geq 0$, we have that
\begin{align}
\label{mochemoche}
\lim_{n \to +\infty} \inf_{P \in \mathcal{Q}^{T}} P\left( \left\{\int_{x_0-\frac{\e}2}^{x_0}U_n''(\cdot,v)dv <-\frac{C}{\e}\right\}\right)=1.
\end{align}
First we prove that \eqref{toprove}  holds true. Fix some $\e>0$ such that $x_0>\e$ and $M > 0$. For all $\o^{T} \in \O^{T}$  $$U_{n}(\o^{T},x_0-\e) = U_{n}(\o^{T},x_0)- \int_{x_0-\e}^{x_0} U'_{n}(\o^{T},u)du.$$ Using that $U'_{n}(\o^{T},\cdot)$ is non-negative and non increasing,  we obtain   that
\begin{align*}
U_{n}(\o^{T},x_0-\e) + \frac{\e}2  U'_{n}\left(\o^{T},x_0-\frac{\e}2\right) & \leq
U_{n}(\o^{T},x_0-\e) + \int_{x_0-\e}^{x_0-\frac{\e}2}U'_{n}(\o^{T},v)dv \leq U_{n}(\o^{T},x_0).
\end{align*}
Now \begin{align*}
U'_{n}\left(\o^{T},x_0-\frac{\e}2\right)&=U'_{n}(\o^{T},x_0)-\int_{x_0-\frac{\e}2}^{x_0} U''_{n}(\o^{T},v)dv
\geq -\int_{x_0-\frac{\e}2}^{x_0} U''_{n}(\o^{T},v)dv
\end{align*}
and all together
\begin{align*}
U_{n}(\o^{T},x_0-\e) & \leq |U_{n}(\o^{T},x_0)|+ \frac{\e}2  \int_{x_0-\frac{\e}2}^{x_0} U''_{n}(\o^{T},v)dv.
\end{align*}
We fix some $\eta>0$ and  show  that there exists some $N_{\eta}>0$ such that  for all $n$
$$\inf_{P \in \mathcal{Q}^{T}} P\left(|U_{n}(\cdot,x_0)| \leq N_{\eta}\right)>1-\frac{\eta}2.$$
{As $\sup_n \|U_{n}(\cdot,x_0)\|_1 <+ \infty,$} \citep[Lemma 13]{DHP11} implies that for all $k \geq 1$ 
$$
\sup_{P \in \mathcal{Q}^{T}}
P\left(|U_{n}(\cdot,x_0)|>  k \right) \leq
\frac1{k}
\sup_{P \in \mathcal{Q}^{T}}E_P\left(|U_{n}(\cdot,x_0)|\right)\leq \frac1{k}
\sup_n \|U_{n}(\cdot,x_0)\|_1.$$
Thus there exists $N_{\eta}>0$ such that
$\sup_{P \in \mathcal{Q}^{T}} P\left(|U_{n}(\cdot,x_0)|>  N_{\eta}\right)<\frac{\eta}2$ for all $n$.  From
\eqref{mochemoche} with  $C=2(N_{\eta}+M)$, there exists $N=N(\eta,M,\e)$ such that
for all $n \geq N$,
\begin{align*}
\inf_{P \in \mathcal{Q}^{T}} &P( U_{n}(\cdot,x_0-\e)  \leq -M) \\
 &\geq \inf_{P \in \mathcal{Q}^{T}}P\left(\{ |U_{n}(\cdot,x_0)| \leq N_{\eta} \} \cap \left\{ \int_{x_0-\frac{\e}2}^{x_0} U''_{n}(\cdot,v)dv <-\frac{2(N_{\eta}+M)}{\e}\right\}\right) \\
 & \geq \inf_{P \in \mathcal{Q}^{T}}P\left(\{ |U_{n}(\cdot,x_0)| \leq N_{\eta} \}\right) + \inf_{P \in \mathcal{Q}^{T}} P\left(\left\{ \int_{x_0-\frac{\e}2}^{x_0} U''_{n}(\cdot,v)dv <-\frac{2(N_{\eta}+M)}{\e}\right\}\right)-1 >1-\eta.
\end{align*}
Thus, \eqref{toprove}  is proved for all $x=x_0- \e>0$. \\
We are left with the proof of  \eqref{mochemoche}. Going back to the assumption of the lemma,  there exists some $N \geq 1$  and a  strictly positive random variable $\l$   such that   $U'_{n}(\cdot,x_0) \geq \l(\cdot)$ for all  $n \geq N$. So  we get that
$$\int_{x_0-\frac{\e}2}^{x_0}U_n''(\cdot,v)dv =-\int_{x_0-\frac{\e}2}^{x_0} U_n'(\cdot,v)r_n(\cdot,v) dv \leq -\l(\cdot)\int_{x_0-\frac{\e}2}^{x_0} r_n(\cdot,v)dv. $$
Thus to prove that  \eqref{mochemoche} holds true, it is enough to show that \begin{align}
\label{ugly}
\lim_{n \to +\infty} \inf_{P \in \mathcal{Q}^{T}}
P\left(\l(\cdot) \int_{x_0 -\frac{\e}2}^{x_0} r_{n}(\cdot,v)dv > \frac{C}{\e}\right)=1.
\end{align}
As for all $ x \in (0,x_{0}]$ $\lim_{n} \rho_{n}(x)=+\infty$, we get that  $\lim_{n \to +\infty} \int_{x_0-\frac{\e}2}^{x_0} \rho_{n}(v)dv=+\infty$ by Fatou's Lemma. 
Now \eqref{ugly} holds true as 
$r_n(\cdot,x) \geq \rho_{n}(x)$, for  $n \geq N$ and $x \in (0,x_{0}]$. 
\end{proof}\\

\begin{remark}
\label{remition}
$1.$ We indicate  why in our proof we cannot use directly the assumption that $\lim_{n \to +\infty} r_{n}(\cdot,x_0)=+\infty$ instead of Assumption \ref{un}. Indeed $\lim_{n \to +\infty} r_{n}(\o^{T},x)=+\infty$ for all $x \in (0,x_0]$, $\o^{T} \in \O^{T}$ together with Fatou's Lemma imply that for all $\o^T \in \Omega^T$, there exists $N_{\o^T}$ such that for all $k \geq N_{\o^T}$,
$\l(\o^{T}) \int_{x_0 -\frac{\e}2}^{x_0} r_{n}(\o^{T},v)dv > \frac{C}{\e}$, which means that
$$\Omega^T=\cup_n \cap_{k \geq n} \left\{\l(\cdot)  \int_{x_0 -\frac{\e}2}^{x_0} r_{k}(\cdot,v)dv>\frac{C}{\e}\right\}$$
and using \citep[Theorem 1]{DHP11} this implies that
\begin{align*}
\lim_{n \to +\infty} \sup_{P \in \mathcal{Q}^{T}} P\left(\l(\cdot) \int_{x_0 -\frac{\e}2}^{x_0} r_{n}(\cdot,v)dv > \frac{C}{\e}\right)=1.
\end{align*}
But this does not imply that \eqref{ugly} holds true, hence we cannot conclude as in the proof of Lemma \ref{L4}  that \eqref{toprove} holds true.\\
$2.$ In the course of the proof of Lemma \ref{L4}  we saw that {if $\sup_n \|U_{n}(\cdot,x_0)\|_1 <+ \infty$}, \eqref{toprove}  is satisfied under other sets of assumptions.
\begin{enumerate}
\item Of course \eqref{mochemoche} implies that \eqref{toprove} holds true.
\item
If  there exists some $N \geq 1$ and a  strictly positive random variable $\l$   such that   $U'_{n}(\cdot,x_0) \geq \l(\cdot)$ for all $n \geq N$, then \eqref{ugly} also implies \eqref{toprove}.
\item If furthermore
 $U''_n(\o^T,\cdot)$ is non decreasing for all $n \geq N$ and $\o^T \in \Omega^T$, then
\begin{align}
\label{bi1}
\lim_{n \to +\infty} \inf_{P \in \mathcal{Q}^{T}}
P\left( \left\{\l(\cdot) r_n(\cdot,x_0)>\frac{2C}{\e^2}\right\}\right)=1.
\end{align}
implies that  \eqref{toprove} holds true.
\end{enumerate}
Indeed for the last assertion, since  for all $n \geq N$ and $\o^T \in \Omega^T$, $U''_n(\o^T,\cdot)$ is non decreasing and $U'_{n}(\o^{T},x_0) \geq \l(\o^{T})$, we get that
$$\int_{x_0-\frac{\e}2}^{x_0}U_n''(\cdot,v)dv \leq \frac{\e}2U_n''(\cdot,x_0)=-\frac{\e}2U_n'(\cdot,x_0)r_n(\cdot,x_0) \leq -\frac{\e}2 \l(\cdot)r_n(\cdot,x_0). $$
Thus \eqref{bi1} implies \eqref{mochemoche} and \eqref{toprove} holds true. Note that power utility functions or exponential utility functions (with random coefficients, see Example \ref{randomuex} for the precise conditions) are examples where $U''_n(\o^T,\cdot)$ is non decreasing for all $n$ and $\o^T \in \Omega^T$.
\end{remark}

\begin{remark}
We revisit briefly the notion of certainty equivalent (see Proposition \ref{certainequivdef})  but for random utility functions in both the uni and multiple-priors framework. Let $G \in \mathcal{W}_{T}^{0}$ such that  $0 \leq G(\cdot) <+\infty$  $\mathcal{Q}^{T}$-q.s. and assume that $U$ is  a utility function verifying Definition \ref{utilitydefdifR} {and such that for all $\omega^{T} \in \Omega^{T}$,  $U \left(\omega^{T},\cdot\right)$ is   concave  and  twice continuously differentiable on $(0,\infty)$. Moreover suppose that}  
$\sup_{P \in \mathcal{Q}^{T}} E_{P} U^{-}(\cdot,y)<+\infty$ for all $y>0$, $E_{P} U^{+}(\cdot,1)<+\infty$  and {$E_{P}| U(\cdot,G(\cdot))|<+\infty$} for all $P \in \mathcal{Q}^{T}$. Then
\begin{itemize}
\item For all $P \in \mathcal{Q}^{T}$, there exists a {unique} constant $e(G, P) \in [0,+\infty)$ such that
\begin{align*}
E_{P} U(\cdot,e(G,P))&=  E_{P} U(\cdot,G(\cdot)).
\end{align*}
\item If furthermore  $G \in \mathcal{W}_{T}^{\infty,+}$, $\sup_{P \in \mathcal{Q}^{T}} E_{P} U^{-}(\cdot,G(\cdot))<\infty$
and  $\inf_{P \in \mathcal{Q}^{T}}E_P U'(\cdot,z)>0$ for all $z>0$,  then there  also exists  a unique
${e}(G) \in [0,||G||_{\infty})$ such that
\begin{align*}
\inf_{P \in \mathcal{Q}^{T}} E_PU(\cdot,{e}(G)) &=  \inf_{P \in \mathcal{Q}^{T}} E_{P} U(\cdot,G(\cdot))
\end{align*}
and in this case, we have that  ${e}(G) \geq \inf_{P \in \mathcal{Q}^{T}} e(G,P)$. We call ${e}(G)$ the multiple-priors  certainty equivalent of $G$.
\end{itemize}
As for Proposition \ref{certainequivdef}, the (omitted) proof relies on a careful application of the intermediate value theorem.
\end{remark}

\bibliography{biblioRomain19}

\begin{thebibliography}{72}
\providecommand{\natexlab}[1]{#1}
\providecommand{\url}[1]{\texttt{#1}}
\expandafter\ifx\csname urlstyle\endcsname\relax
  \providecommand{\doi}[1]{doi: #1}\else
  \providecommand{\doi}{doi: \begingroup \urlstyle{rm}\Url}\fi

\bibitem[Acciaio and Penner(2011)]{AcPe11}
B.~Acciaio and I.~Penner.
\newblock \emph{Dynamic convex risk measures in Advanced Mathematical Methods
  for Finance}, chapter~1.
\newblock Springer-Verlag, Berlin, 2011.

\bibitem[Acciaio et~al.(2013)Acciaio, Beiglbock, Penkner, and
  Schachermayer]{ABPW13}
B.~Acciaio, M.~Beiglbock, F.~Penkner, and W.~Schachermayer.
\newblock A model-free version of the fundamental theorem of asset pricing and
  the super-replication theorem.
\newblock \emph{Mathematical Finance}, 26\penalty0 (2):\penalty0 233--251,
  2013.

\bibitem[Aliprantis and Border(2006)]{Hitch}
C.~D. Aliprantis and K.~C. Border.
\newblock \emph{Infinite Dimensional Analysis~: A Hitchhiker's Guide}.
\newblock Grundlehren der Mathematischen Wissenschaften [Fundamental Principles
  of Mathematical Sciences]. Springer-Verlag, Berlin, 3rd edition, 2006.

\bibitem[Artzner et~al.(1999)Artzner, Delben, Eber, and Heath]{Art99}
P.~Artzner, F.~Delben, J.~M. Eber, and D.~Heath.
\newblock Coherent measures of risk.
\newblock \emph{Mathematical Finance}, 9:\penalty0 203--227, 1999.

\bibitem[Avellaneda et~al.(1996)Avellaneda, Levy, and Paras]{AP95}
M.~Avellaneda, A.~Levy, and A.~Paras.
\newblock Pricing and hedging derivatives securities in markets with uncertain
  volatilities.
\newblock \emph{Applied Mathematical Finance}, 2\penalty0 (2):\penalty0 73--88,
  1996.

\bibitem[Bank et~al.(2016)Bank, Dolinsky, and Gokay]{Dol16}
P.~Bank, Y.~Dolinsky, and S.~Gokay.
\newblock Super-replication with nonlinear transaction costs and volatility
  uncertainty.
\newblock \emph{Annals of Applied Probability}, 26\penalty0 (3):\penalty0
  1698--1726, 2016.

\bibitem[Bartl(2019)]{Bart16}
D.~Bartl.
\newblock Exponential utility maximization under model uncertainty for
  unbounded endowments.
\newblock \emph{Annals of Applied Probability}, 29\penalty0 (1):\penalty0
  577--612, 2019.

\bibitem[Beiglbockock et~al.(2013)Beiglbockock, Henry-Labordere, and
  Penkner]{BeiHLPen13}
M.~Beiglbockock, P.~Henry-Labordere, and F.~Penkner.
\newblock Model-independent bounds for option prices: a mass transport
  approach.
\newblock \emph{Finance Stoch.}, 17\penalty0 (3):\penalty0 477--501, 2013.

\bibitem[Bensaid et~al.(1992)Bensaid, Lesne, Pages, and Scheinkman]{Ben91}
B.~Bensaid, J.~P. Lesne, H.~Pages, and J.~Scheinkman.
\newblock Derivative asset pricing with transaction costs.
\newblock \emph{Mathematical Finance}, 2\penalty0 (2):\penalty0 63--86, 1992.

\bibitem[Bertsekas and Shreve(2004)]{BS}
D.~P. Bertsekas and S.~Shreve.
\newblock \emph{Stochastic Optimal Control: The Discrete-Time Case}.
\newblock Athena Scientific, 2004.

\bibitem[Bielecki et~al.(2016)Bielecki, Cialenco, and Pitera]{Bie16}
T.~R. Bielecki, I.~Cialenco, and M.~Pitera.
\newblock A survey of time consistency of dynamic risk measures and dynamic
  performance measures in discrete time: Lm-measure perspective.
\newblock \emph{Probability, Uncertainty and Quantitative Risk}, 2\penalty0
  (3), 2016.

\bibitem[Blanchard and Carassus(2018)]{BC16}
R.~Blanchard and L.~Carassus.
\newblock Multiple-priors investment in discrete time for unbounded utility
  function.
\newblock \emph{Annals of Applied Probability}, 88\penalty0 (2):\penalty0
  241--281, 2018.

\bibitem[Blanchard and Carassus(2019)]{BC19}
R.~Blanchard and L.~Carassus.
\newblock No arbitrage with multiple priors.
\newblock \emph{arxiv}, 2019.

\bibitem[Bouchard(2000)]{bouchard-these}
B.~Bouchard.
\newblock \emph{Stochastic control and applications in finance.}
\newblock PhD thesis, Universit\'e Paris 9, 2000.

\bibitem[Bouchard and Nutz(2015)]{BN}
B.~Bouchard and M.~Nutz.
\newblock Arbitrage and duality in nondominated discrete-time models.
\newblock \emph{Annals of Applied Probability}, 25\penalty0 (2):\penalty0
  823--859, 2015.

\bibitem[Burzoni et~al.(2016)Burzoni, Frittelli, and Magis]{ClassS}
M.~Burzoni, M.~Frittelli, and M.~Magis.
\newblock Universal arbitrage aggregator in discrete-time markets under
  uncertainty.
\newblock \emph{Finance and Stochastics}, 20\penalty0 (1-50), 2016.

\bibitem[Carassus and R\'asonyi(2006)]{CR06}
L.~Carassus and M.~R\'asonyi.
\newblock Convergence of utility indifference prices to the superreplication
  price.
\newblock \emph{Mathematical Methods of Operations Research}, 64:\penalty0
  145--154, 2006.

\bibitem[Carassus and R\'asonyi(2007{\natexlab{a}})]{CR07b}
L.~Carassus and M.~R\'asonyi.
\newblock Convergence of utility indifference prices to the superreplication
  price : the whole real line case.
\newblock \emph{Acta Applicandae Mathematicae}, 96\penalty0 (119-135),
  2007{\natexlab{a}}.

\bibitem[Carassus and R\'asonyi(2007{\natexlab{b}})]{cr}
L.~Carassus and M.~R\'asonyi.
\newblock Optimal strategies and utility-based prices converge when agents'
  preferences do.
\newblock \emph{Mathematics of Operations Research}, 32:\penalty0 102--117,
  2007{\natexlab{b}}.

\bibitem[Carassus and R\'asonyi(2011)]{CR11}
L.~Carassus and M.~R\'asonyi.
\newblock Risk-averse asymptotics for reservation prices.
\newblock \emph{Annals of Finance}, 7\penalty0 (3):\penalty0 375--387, 2011.

\bibitem[Carassus and Vargiolu(2018)]{CV17}
L.~Carassus and T.~Vargiolu.
\newblock Super-replication price: it can be ok.
\newblock \emph{Esaim: proceedings and surveys}, 64, 2018.

\bibitem[Carassus et~al.(2019)Carassus, Obl\'oj, and Wiesel]{COW}
L.~Carassus, J.~Obl\'oj, and J.~Wiesel.
\newblock The robust superreplication problem: a dynamic approach.
\newblock \emph{SIAM J. Financial Mathematics}, 10\penalty0 (4):\penalty0
  907--941, 2019.

\bibitem[Carmona(2009)]{Car09}
R.~Carmona.
\newblock \emph{Indifference Pricing : Theory and Applications}.
\newblock Princeton University Press, 2009.

\bibitem[Cerreia~Vioglio et~al.(2011)Cerreia~Vioglio, Maccheroni, Marinacci,
  and Montrucchio]{Vio11}
S.~Cerreia~Vioglio, F.~Maccheroni, M.~Marinacci, and L.~Montrucchio.
\newblock Uncertainty averse preferences.
\newblock \emph{Journal of Economic Theory}, 146\penalty0 (4):\penalty0
  1275--1330, 2011.

\bibitem[Cerreia~Vioglio et~al.(2015)Cerreia~Vioglio, Maccheroni, and
  Marinacci]{Vio15}
S.~Cerreia~Vioglio, F.~Maccheroni, and M.~Marinacci.
\newblock Put--call parity and market frictions.
\newblock \emph{Journal of Economic Theory}, 157, 2015.

\bibitem[Cheridito et~al.(2017)Cheridito, Kupper, and Tangpi]{Cher17}
P.~Cheridito, M.~Kupper, and L.~Tangpi.
\newblock Duality formulas for robust pricing and hedging in discrete time.
\newblock \emph{Banach J. Math. Anal.}, 11 (1):\penalty0 72--89, 2017.

\bibitem[Cohen(2012)]{Coh12}
S.~Cohen.
\newblock Quasi-sure analysis, agregation and dual representations of sublinear
  expectations in general spaces.
\newblock \emph{Electronic Journal of Probability}, 17\penalty0 (62), 2012.

\bibitem[Cox and Obl\'oj(2011{\natexlab{a}})]{CoOb11}
A.M.G Cox and J.~Obl\'oj.
\newblock Robust hedging of double touch barrier options.
\newblock \emph{SIAM J. Finan. Math.}, 2:\penalty0 141--182,
  2011{\natexlab{a}}.

\bibitem[Cox and Obl\'oj(2011{\natexlab{b}})]{CoOb112}
A.M.G Cox and J.~Obl\'oj.
\newblock Robust pricing and hedging of double no-touch options.
\newblock \emph{Finance and Stochastics}, 15\penalty0 (3):\penalty0 573--605,
  2011{\natexlab{b}}.

\bibitem[Cox et~al.(1979)Cox, Ross, and Rubistein]{CRR79}
J.C. Cox, S.A Ross, and M.~Rubistein.
\newblock Option pricing: a simplified approach.
\newblock \emph{Journal of Financial Economics}, 7\penalty0 (229-264), 1979.

\bibitem[Cvitani\'c and Karatzas(1992)]{CvKa92}
J.~Cvitani\'c and I.~Karatzas.
\newblock Hedging contingent claims with contrained portfolio.
\newblock \emph{Annals of Applied Probability}, 2\penalty0 (4):\penalty0
  767--818, 1992.

\bibitem[Cvitani\'c et~al.(1999)Cvitani\'c, Pham, and Touzi.]{CvPT99}
J.~Cvitani\'c, H.~Pham, and N.~Touzi.
\newblock Super-replication in stochastic volatility models under portfolio
  constraints.
\newblock \emph{Journal of Applied Probability}, 2\penalty0 (523-545), 1999.

\bibitem[Dalang et~al.(1990)Dalang, Morton, and Willinger]{dmw}
R.~C. Dalang, A.~Morton, and W.~Willinger.
\newblock Equivalent martingale measures and no-arbitrage in stochastic
  securities market models.
\newblock \emph{Journal Stochastics and Stochastic Reports.}, 29:\penalty0
  185--201, 1990.

\bibitem[Davis and Hobson(2007)]{DaHo07}
M.H.A. Davis and D.~Hobson.
\newblock The range of traded option prices.
\newblock \emph{Mathematical Finance}, 17\penalty0 (1):\penalty0 1--14, 2007.

\bibitem[Delbaen and Schachermayer(2006)]{DelSch05}
F.~Delbaen and W.~Schachermayer.
\newblock \emph{The Mathematics of Arbitrage}.
\newblock Springer Finance, 2006.

\bibitem[Delbaen et~al.(2002)Delbaen, Grandits, Rheinl\"ander, Samperi,
  Schweizer, and Stricker.]{sixauthor}
F.~Delbaen, P.~Grandits, T.~Rheinl\"ander, D.~Samperi, M.~Schweizer, and Ch.
  Stricker.
\newblock Exponential hedging and entropic penalties.
\newblock \emph{Mathematical Finance}, 12:\penalty0 99--123, 2002.

\bibitem[Denis and Martini(2006)]{DM06}
L.~Denis and C.~Martini.
\newblock A theoretical framework for the pricing of contingent claims in the
  presence of model uncertainty.
\newblock \emph{Annals of Applied Probability}, 16(2):\penalty0 827--852, 2006.

\bibitem[Denis et~al.(2011)Denis, Hu, and Peng]{DHP11}
L.~Denis, M.~Hu, and S.~Peng.
\newblock Function spaces and capacity related to a sublinear expectation:
  application to {G}-brownian motion paths.
\newblock \emph{Potential Analysis}, 34 (2)\penalty0 (139-161), 2011.

\bibitem[Dolinsky and Soner(2014)]{DolSon14}
Y.~Dolinsky and H.~M. Soner.
\newblock Martingale optimal transport and robust hedging in continuous time.
\newblock \emph{Probability Theory and Related Fields}, 160\penalty0
  (1):\penalty0 391--427, 2014.

\bibitem[{E}l {K}aroui and Quenez(1991)]{ElQu91}
N.~{E}l {K}aroui and M.-C Quenez.
\newblock Dynamic programming and pricing of contingent claims in incomplete
  markets.
\newblock \emph{SIAM J. Control Optim.}, 33\penalty0 (1):\penalty0 29--66,
  1991.

\bibitem[Elsberg(1961)]{El61}
D.~Elsberg.
\newblock Risk, ambiguity, and the savage axioms.
\newblock \emph{Quarterly Journal of Economics}, 75\penalty0 (4):\penalty0
  643--669, 1961.

\bibitem[Epstein and Ji(2014)]{EJ13}
L.~G. Epstein and S.~Ji.
\newblock Ambigous volatility, possibility and utility in continuous time.
\newblock \emph{Journal of Mathematical Economic}, 50:\penalty0 269--282, 2014.

\bibitem[Epstein and Schneider(2003)]{ES03}
L.~G. Epstein and M.~Schneider.
\newblock Recursive multiple-priors.
\newblock \emph{Journal of Economic Theory}, 113\penalty0 (1):\penalty0 1--31,
  2003.

\bibitem[F{\"o}llmer and Schied(2002)]{fs}
H.~F{\"o}llmer and A.~Schied.
\newblock \emph{Stochastic Finance: An Introduction in Discrete Time}.
\newblock Walter de Gruyter \& Co., Berlin, 2002.

\bibitem[Giammarino and Barrieu(2013)]{Bau13}
F.~Giammarino and P.~M. Barrieu.
\newblock Indifference pricing with uncertainty averse preferences.
\newblock \emph{Journal of Mathematical Economics}, 49\penalty0 (1), 2013.

\bibitem[Gilboa(2009)]{Gilboa}
I.~Gilboa.
\newblock \emph{Theory of Decision under Uncertainty}.
\newblock Econometric Society Monographs, 2009.

\bibitem[Gilboa and Schmeidler(1989)]{Gilb}
I.~Gilboa and D.~Schmeidler.
\newblock Maxmin expected utility with non-unique prior.
\newblock \emph{Journal of Mathematical Economics}, 18(2):\penalty0 141--153,
  1989.

\bibitem[Hansen and Sergent(2001)]{Han01}
L.~P. Hansen and T.~J. Sergent.
\newblock Robust control and model uncertainty.
\newblock \emph{American Economic Review}, 91, 2001.

\bibitem[Harrison and Kreps(1979)]{Hakr79}
J.~M. Harrison and D.~M. Kreps.
\newblock Martingales and arbitrage in multiperiod securities markets.
\newblock \emph{Journal of Economic Theory}, 20\penalty0 (3):\penalty0
  381--408, 1979.

\bibitem[Harrison and Pliska(1981)]{HaPL81}
J.~M. Harrison and S.~R. Pliska.
\newblock Martingales and stochastic integrals in the theory of continuous
  trading.
\newblock \emph{Stochastic Processes and their Applications}, 11\penalty0
  (3):\penalty0 215--260, 1981.

\bibitem[Henderson(2002)]{VH02}
V.~Henderson.
\newblock Valuation of claims on nontraded assets using utility maximization.
\newblock \emph{Mathematical Finance}, 12:\penalty0 351--373, 2002.

\bibitem[Hobson(1998)]{Ho982}
D.~Hobson.
\newblock Robust hedging of the lookback option.
\newblock \emph{Finance and Stochastics}, 2:\penalty0 329--347, 1998.

\bibitem[Hodges and Neuberger(1989)]{H89}
R.~Hodges and K.~Neuberger.
\newblock Optimal replication of contingent claims under transaction costs.
\newblock \emph{Rev. Futures Mkts.}, 8:\penalty0 222--239, 1989.

\bibitem[Kahneman and Tversky(1979)]{kt}
D.~Kahneman and A.~Tversky.
\newblock Prospect theory: An analysis of decision under risk.
\newblock \emph{Econometrica}, 47:\penalty0 263--291, 1979.

\bibitem[Kallblad et~al.(2017)Kallblad, Obloj, and Zariphopoulou]{KOZ17}
S.~Kallblad, J.~Obloj, and T.~Zariphopoulou.
\newblock Dynamically consistent investment under model uncertainty: the robust
  forward criteria.
\newblock \emph{Finance and Stochastics}, To appear, 2017.

\bibitem[Knight(1921)]{Kni}
F.~Knight.
\newblock \emph{Risk, Uncertainty, and Profit}.
\newblock Boston, MA: Hart, Schaffner Marx; Houghton Mifflin Co, 1921.

\bibitem[Kreps(1981)]{Kr81}
D.~M. Kreps.
\newblock Arbitrage and equilibrium in economies with infinitely many
  commodities.
\newblock \emph{Journal of Mathematical Economics}, 8\penalty0 (1):\penalty0
  15--35, 1981.

\bibitem[Lyons(1995)]{Ly95}
F.~Lyons.
\newblock Uncertain volatility and the risk-free synthesis of derivatives.
\newblock \emph{Journal of Applied Finance}, 2:\penalty0 117--133, 1995.

\bibitem[Maccheroni et~al.(2006)Maccheroni, Marinacci, and Rustichini]{Macc06}
F.~Maccheroni, M.~Marinacci, and A.~Rustichini.
\newblock Ambiguity aversion, robustness, and the variational representation of
  preferences.
\newblock \emph{Econometrica}, 74\penalty0 (6):\penalty0 1447--1498, 2006.

\bibitem[Mehra and Prescott(1985)]{MP85}
R.~Mehra and E.C. Prescott.
\newblock The equity premium: A puzzle.
\newblock \emph{Journal of Monetary Economics}, 15\penalty0 (2):\penalty0
  145--161, 1985.

\bibitem[Monoyios(2004)]{MM04}
M.~Monoyios.
\newblock Performance of utility-based strategies for hedging basis risk.
\newblock \emph{Quantitative Finance}, 4\penalty0 (244-255), 2004.

\bibitem[Musiela and Zariphopoulou(2005)]{Mu05}
M.~Musiela and Th. Zariphopoulou.
\newblock \emph{Indifference Pricing}, chapter Backward and forward utilities
  and the associated indifference pricing systems: The case study of the
  binomial model.
\newblock Princeton University Press, 2005.

\bibitem[Nutz(2016)]{Nutz}
M.~Nutz.
\newblock Utility maximisation under model uncertainty in discrete time.
\newblock \emph{Mathematical Finance}, 26\penalty0 (2):\penalty0 252--268,
  2016.

\bibitem[Peng(2011)]{pg11}
S.~Peng.
\newblock Backward stochastic differential equation, nonlinear expectation and
  their application.
\newblock In R.~Bhatia, editor, \emph{Proceeding of the International Congress
  of Mathematics}, volume~1, pages 393--432. World Scientific, Singapore, 2011.

\bibitem[Pratt(1964)]{Pr65}
J.~Pratt.
\newblock Risk aversion in the small and in the large.
\newblock \emph{Econometrica}, 32:\penalty0 122--135, 1964.

\bibitem[R\'asonyi and Stettner(2006)]{RS06}
M.~R\'asonyi and L.~Stettner.
\newblock On the existence of optimal portfolios for the utility maximization
  problem in discrete time financial models.
\newblock \emph{In: Kabanov, Y.; Lipster, R.; Stoyanov,J. (Eds), From
  Stochastic Calculus to Mathematical Finance, Springer.}, pages 589--608,
  2006.

\bibitem[Riedel(2009)]{Rie09}
F.~Riedel.
\newblock Optimal stopping with multiple priors.
\newblock \emph{Econometrica}, 77\penalty0 (3):\penalty0 857--908, 2009.

\bibitem[Riedel(2011)]{Rie11}
F.~Riedel.
\newblock Finance without probabilistic prior assumptions.
\newblock \emph{ArXiv}, 2011.

\bibitem[Rouge and {E}l {K}aroui(2000)]{ElKRou00}
R.~Rouge and N.~{E}l {K}aroui.
\newblock Pricing via utility maximization and entropy.
\newblock \emph{Mathematical Finance}, 10\penalty0 (2):\penalty0 259--276,
  2000.

\bibitem[Savage(1954)]{Sav}
L.~Savage.
\newblock The foundations of statistics.
\newblock \emph{Wiley, New York}, 1954.

\bibitem[Soner et~al.(2011)Soner, Touzi, and Zhang]{SoToZa11}
H.~Mete Soner, N.~Touzi, and J.~Zhang.
\newblock Quasi-sure stochastic analysis through aggregation.
\newblock \emph{Electronic Journal of Probability}, 16\penalty0 (67):\penalty0
  1844--1879, 2011.

\bibitem[von Neumann and Morgenstern(1947)]{vNM}
J.~von Neumann and O.~Morgenstern.
\newblock \emph{Theory of games and economic behavior}.
\newblock Princeton University Press, 1947.

\end{thebibliography}

\end{document}